\documentclass[unsortedaddress,onecolumn, nofootinbib,10pt]{revtex4}
\usepackage[a4paper,bindingoffset=0.2in, left=0.5in,right=0.8in,top=1in,bottom=1in, footskip=.25in]{geometry}
\usepackage[utf8]{inputenc}
\usepackage{amsmath,empheq,color}
\usepackage{amsfonts}
\usepackage{amsthm}
\usepackage{amssymb}
\usepackage{graphicx}
\usepackage{hyperref}
\usepackage[normalem]{ulem}
\usepackage{epigraph}
\usepackage{mathrsfs}
\usepackage{bbm}
\usepackage{wrapfig}
\usepackage{lineno}
\usepackage{scalerel,stackengine}
\usepackage[T1]{fontenc}
\usepackage{epigraph}
\usepackage{cancel}
\usepackage{soul}
\usepackage{ulem}
\usepackage{hyperref}
\usepackage{amsmath,amssymb,amsfonts,amsthm,amscd}
\usepackage{graphicx}
\usepackage{enumerate}
\usepackage{colordvi}
\usepackage{units}
\usepackage{epsfig}
\usepackage{natbib}
\usepackage{enumerate}
\usepackage{colordvi}
\usepackage{multirow}
\usepackage{afterpage}
\usepackage{subfig}
\usepackage[utf8]{inputenc}
\usepackage{xcolor}
\usepackage{tikz-cd}

\usepackage[textwidth=1.7cm,textsize=small]{todonotes}

\newtheorem{theorem}{Theorem}
\newtheorem{remark}{Remark}
\newtheorem{corollary}{Corollary}
\newtheorem{example}{Example}
\newtheorem{definition}{Definition}
\newtheorem{proposition}{Proposition}

\def\be{\begin{equation}}
\def\ee{\end{equation}}
\def\ba{\begin{eqnarray}}
\def\ea{\end{eqnarray}}

\def\H{\mathscr{H}}
\def\D{\mathbf{D}}

\def\s{\mathbf{s}}

\def\X{\mathbf{X}}

\def\R{\mathbb R}

\newcommand{\fe}{\varphi}
\newcommand{\inv}{^{-1}}
\newcommand{\del}{\partial}
\newcommand{\e}{\varepsilon}
\newcommand{\C}{\mathbb{C}}

\usepackage{todonotes}
\begin{document}

%\title{Dynamical Similarities, Contact Reductions and Evolving Couplings}
\title{Scaling Symmetries, Contact Reduction and Poincar\'e's dream}
%\todo{check the other titles commented in the .tex}

\author{Alessandro Bravetti}
\email{alessandro.bravetti@iimas.unam.mx} 
\affiliation{Instituto de Investigaciones en Matem\'aticas Aplicadas y en Sistemas, 
Universidad Nacional Aut\'onoma de M\'exico, A.~P.~70543, M\'exico, DF 04510, Mexico}
\author{Connor Jackman}
\email{connor.jackman@cimat.mx} 
\affiliation{CIMAT, A.P. 402, Guanajuato, Gto. 36000, M\'exico}

\author{David Sloan}
\email{d.sloan@lancaster.ac.uk}
\affiliation{Department of Physics, Lancaster University, Lancaster UK}

\begin{abstract}
A symplectic Hamiltonian system admitting a 
%dynamical similarity 
scaling symmetry
can be reduced to an equivalent contact Hamiltonian system 
in which some physically-irrelevant degree of freedom has been removed. 
As a consequence, one obtains an equivalent description for the same physical phenomenon, but with fewer inputs needed,
%\textcolor{blue}{(one gets rid precisely of the input related with the scale) (Remove?)}, 
thus realizing `Poincar\'e's dream' of a scale-invariant 
description of the universe.
%\textcolor{blue}{Moreover, it has been shown that in many cases this scale-invariant description
%can be continued through the
%singularities of the original symplectic system, with potentially groundbreaking effects for our views on the universe. (True, but not directly this paper)}

This work is devoted to a thorough analysis of the mathematical framework behind such reductions. 
We show that 
generically such reduction is possible and the reduced (fundamental) system is a contact Hamiltonian system.
The price to pay for this level of generality is that one is compelled to
consider the coupling constants appearing in the original Hamiltonian
as part of the dynamical variables of a lifted system. 
This however has the added advantage of removing the hypothesis of
the existence of a scaling symmetry for the original system at all, without breaking
the sought-for reduction in the number of inputs needed.
Therefore %we conclude that
a large class of Hamiltonian (resp.~Lagrangian)
theories can be reduced to 
scale-invariant
contact Hamiltonian (resp.~Herglotz variational) theories.
\end{abstract}

\maketitle

%\keywords{Hamiltonian Mechanics, scaling symmetries contact geometry}

\section{Introduction}

Symmetry plays an important role in physics. 
When considering observations of measurable quantities, 
the action of symmetry should be carefully considered in its effect not only on the quantity being 
measured but also upon the apparatus through which the measurement is made.
This is common when considering most physical symmetries. 
On applying a Galilean transformation to a system -- for example transforming reference frames -- 
we argue invariance of observation based on applying the transformation to both the observer and the observed quantity. 
Absolute motion should not be observable as it is not invariant under Galilean transformations. Physically meaningful quantities and their equations of motion are those invariant under such symmetry transformations.
%Absolute motion should not be observable. The equations of motion of a system that is invariant under Galilean transformations 
%form a \textcolor{red}{closed algebra}\connorscomment{let's omit this language as above (we do not define or use in what follows)} that is independent of absolute position or velocity. 

It has been argued, at least since Poincar\'e (see eg.~\cite{gryb2021scale} for a detailed discussion), that the same considerations should be applied to scaling symmetries and the absolute scale of a system. In \cite{poincare2003science}, pg.~94, Poincar\'e imagines that overnight all dimensions in the universe have been scaled, and argues that there will be no detectable change due to both the observed system (the universe) and any possible measuring apparatus (eg one's height), having undergone a scaling by the same factor.
Accordingly, we will refer to writing the equations of motion of a system in terms of scale-invariant 
quantities as {\em Poincar\'e's dream}.

The counter to this argument was long held to be the fact that physical constants can be used to give an overall scale based on, 
for example, the Planck length. 
However in recent work it has been suggested that if these are included not as given facts about the universe, 
but rather as physical observations of quantities which must be measured, 
such a scaling symmetry can be restored~\cite{sloan2021scale}.

It should be clear that in any scale-invariant description of the physical world one would need one less piece of information with 
respect to the standard descriptions that include reference to some (unmeasurable) absolute scale, the datum
being removed being precisely the one corresponding to the scale.
Moreover, such scale reductions contain valuable information on certain singularities, eg collision singluarities, 
of the original systems and this has been exploited in several works, \cite{Koslowski:2016hds,Mercati:2019cbn,sloan2019scalar,mercati2021through}, 
to find that
the scale-invariant description is freed from the singularities of the original 
dynamics~
(although care must be taken, since this is not always necessarily the case~\cite{mercati2021total}).
As well, one finds interesting dynamical features such as a dissipative-like behavior, which can provide a natural 
origin for the observed arrow of time~\cite{barbour2014identification,sloan2018dynamical,gryb2021new}.
Therefore Poincar\'e's dream is well worth pursuing after all.

Despite the great amount of work %that has been devoted 
in this direction, still %to this question and the compelling arguments in favor
%of Poincar\'e's dream, still
the treatment of scaling symmetries and of their corresponding reductions 
so far has been based on a
case-by-case study (with the exeption of~\cite{sloan2018dynamical})
and with some of the methods being quite ad-hoc; 
for instance, in most situations
the choice of scaling function and selection of scale invariant quantities %to form the dynamical algebra
has been given without justification. 
While in simple systems such as the Kepler problem the choice of $r$ as scaling may indeed appear natural, 
as $r$ represents the physical distance separation between the two bodies in the problem, 
%However, 
when we move beyond this even simply to the $n$-body problem the choice is not so apparent 
-- one could use the mean separation of particles, or the largest separation, or some combination of any variables that scale in 
a similar way. 
Further, even in the simple Kepler problem, one could have opted to use $p_r$, the momentum conjugate to $r$, instead of $r$. %sin the Kepler problem. 
Would the results ensuing from these two different choices be related somehow? 
Although it is not difficult to show directly that analogous results hold  in this particular case, 
this highlights the arbitrary nature of such choices. 

In this work we will address for the first time the study of scaling symmetries in Hamiltonian systems 
and their related contact reductions in full generality.
Our main results (Theorems~\ref{th:CHS} and~\ref{thm:lifted_red}) will 
establish that any Hamiltonian (resp. variational) 
theory admitting a scaling symmetry can be reduced to a dynamically equivalent 
scale-invariant
contact Hamiltonian (resp.~Herglotz variational) theory, thus illustrating that Poincar\'e's dream
is indeed feasible. Surprisingly, we will see that this is possible even in the absence of an a-priori evident scaling symmetry,
provided some less stringent conditions are satisfied.
Moreover, we will consider explicitly the case in which one has chosen a certain scaling function
and show how to construct the resulting equations of motion in terms of scale-invariant quantities %a dynamical algebra of invariants from it,
%thereby also proving %(Theorem~\ref{th:CHS}). 
% the equivalence
%between different choices of scaling functions and corresponding reductions 
(Corollary~\ref{cor:coords}).
Finally, a collection of physically-relevant examples will be provided throughout the manuscript 
in order to illustrate the main points
and constructions.

The paper is laid out as follows:
in Section~\ref{sec:contactreduction} we give a precise definition of a scaling symmetry 
and prove that reducing a symplectic system with a scaling symmetry in general yields two contact systems on different regions,
which coincide if and only if the degree of the scaling symmetry is one (Theorem~\ref{th:CHS} and Corollary~\ref{cor:D1}).
Besides, in Corollary~\ref{cor:coords} we prove an invariance (up to time reparametrizations) 
of the reduced system with respect to the choice of the scaling function.
Finally, in Section~\ref{sec:globalreduction} 
we show that  Hamiltonian systems which do not exhibit scaling symmetries can be lifted into a broader space of models in which the symmetry is present, 
by promoting the physical constants (couplings) to observable variables within the system. 
In this way one may always reduce the description to an equivalent theory based on a contact Hamiltonian
(or Herglotz variational) theory on a reduced space.
Despite the dynamical equivalence among all such theories, the contact Hamiltonian one, involving less elements for its complete description,
has more explanatory power (Proposition~\ref{thm:globalreduction} and Theorem~\ref{thm:lifted_red}).
We conclude in Section~\ref{sec:conclusions} with an outlook on future work and applications.
To have the paper as self-contained as possible and in order to further motivate our constructions, three appendices are also included: 
in Appendix~\ref{app:ham_cont} we introduce the symplectic and contact Hamiltonian systems that we consider, 
and briefly summarize some known relationship betweem them. In Appendix~\ref{app:Herglotz} we review the variational approach of Herglotz
to contact systems and re-state our main results in terms of this formulation. In Appendix~\ref{app:Kepler} we provide some examples
of well-known constructions in celestial mechanics for which our formalism provides a unifying principle.

%\section{Geometric setting}
\section{Scaling symmetries, contact reduction and symplectification}\label{sec:contactreduction}

In this section we describe how a special type of symmetry, called `scaling symmetry', 
provides a new link between symplectic mechanics and contact mechanics 
(see Appendix~\ref{app:ham_cont} for the relevant definitions and our sign conventions). 
Such symmetries have been shown to be present in cosmology, 
the Kepler problem and a host of other physical contexts~\cite{sloan2018dynamical,sloan2019scalar,sloan2021new,sloan2021scale,gryb2021scale}. 

\subsection{Scaling symmetries}\label{subsec:scalingsymm}
We start with a general definition, comprising all those transformations that reparametrize the dynamics without altering 
the (unparametrized) orbits.

\begin{definition}\label{def:dynsim}\rm
A \emph{dynamical similarity} of a vector field $X\in \mathfrak X(M)$ is a vector field $Y\in \mathfrak X(M)$ such
that $[Y, X] = f X$, for some (in general non-constant) function $f:M\rightarrow\R$ 
($Y$ is also called a \emph{conformal symmetry}).
\end{definition}
In other words, a dynamical similarity is a symmetry of the one-dimensional distribution, or line field,
generated by $X$. When $Y\ne 0$, such symmetries give rise to a `local reduction of order'~\cite{arnold2012geometrical},~pg.~7, 
in describing integral curves of $X$.

Among dynamical similarities, there is a special class appearing in Hamiltonian systems that is often found in 
physical applications, the so-called scaling symmetries.

	\begin{definition}\label{def:scaling}\rm
	$\D\in\mathfrak{X}(M)$ is a {\em scaling symmetry} of degree $\Lambda\in\R$ for the Hamiltonian system $(M, \omega, H)$ if
		\begin{itemize}
		\item[i)] $L_{\D}\omega = \omega$ %\textcolor{red}{\st{($\D$ is a Liouville vector field)}}
		\item[ii)] $L_{\D}H = \Lambda H$.
		\end{itemize}
	\end{definition}

Note that with this definition 
	\begin{equation}\label{eq:scalingXH}
	[\D,X_{H}]=(\Lambda-1)X_{H}
	\end{equation}
and therefore scaling symmetries of symplectic Hamiltonian systems are a very special 
case of dynamical similarities. Further, 
note that for any conserved quantity generated by the vector field $J$, we have $[X_H,J]=0$ hence 
from~\eqref{eq:scalingXH} and
the Jacobi identity
\begin{equation}
    [X_H,[\D,J]] = [J,[\D,X_H]] + [\D,[X_H,J]]=0\,,
\end{equation}
we see that dynamical similarities in general (and scaling symmetries in particular) map conserved quantities onto conserved quantities.

\begin{example}[Kepler scalings]\label{ex:Kepler}\rm
Our archetype motivating this definition is the map between solutions to the (planar) Kepler problem,  with 
    \be\label{eq:KeplerSetting1}
    M=T^{*}(\C\backslash 0)  = (\C\backslash 0)\times \C\ni (q,p), 
    \qquad \omega= dp\wedge dq = d( p\cdot  dq),
    \qquad \text{and} \qquad
    H_{K}=\frac{{|p|}^{2}}{2}-\frac{1}{{|q|}}\,.
    \ee
    In this case it is well-known that the \emph{Kepler scalings} 
    $$q\rightarrow \lambda^{2}q 
    \qquad 
    \text{and} \qquad
    p\rightarrow \lambda^{-1}p\,,
    $$
for $\lambda\in \R\backslash 0$, send orbits into orbits. 
Considering the generator of these scalings in the phase space, 
    \be\label{eq:DKKepler1}
    %\label{eq:DK}
    \D_{K}=2q\cdot \partial_{q}-p\cdot \partial_{p}\,,
    \ee
one has
$L_{\D_{K}}\omega =\omega$ and $L_{\D_{K}}H_{K}=-2H_{K}$, that is, $\D_{K}$ is a scaling symmetry of degree $-2$.
By equation~\eqref{eq:scalingXH}, it follows that the time parameter should be rescaled as $t\rightarrow \lambda^{3}t$ in order to match the corresponding 
parametrizations (`Kepler's third law').
\end{example}

Before exploring further the role of scaling symmetries in Hamiltonian systems, the following two related remarks are at order:
\begin{remark}\label{rmrk:freedom1}\rm
The set of vector fields satisfying $L_{\D}\omega=\omega$ (the {\em Liouville vector fields})
form an affine space, directed by the {\em symplectic vector fields}: %(or infinitesimal canonical/symplectic transformations): 
$Y\in\mathfrak{X}(M)$ such that $L_Y\omega = 0$ (equivalently $i_Y\omega$ is closed). 
\end{remark}
\begin{remark}\label{rmrk:freedom2}\rm
%\davescomment{Should we go one step further here and note that any two scalings are related up to this freedom?}
A scaling symmetry need not be unique, in fact for any two scaling symmetries, $\D, \D'$, of $H$ one has $i_{\D'}\omega = i_\D\omega + \alpha$ for some closed 1-form $\alpha$ satisfying $i_{X_H}\alpha = (\Lambda' - \Lambda) H$, where $\Lambda', \Lambda$ are the degrees of $\D', \D$. Conversely, any closed 1-form with $i_{X_H}\alpha = (\Lambda' - \Lambda) H$ determines another scaling symmetry of degree $\Lambda'$ through $i_{\D'}\omega = i_\D\omega + \alpha$.
For example, given a first integral $F$ of $H$, and scaling symmetry $\D$, we have that 
$\D+X_{F}$ is also a scaling symmetry of the same degree as $\D$. In particular, one may always add $X_{H}$ to $\D$.
\end{remark}	

Let us conclude this section by 
introducing a re-characterization of scaling symmetries that is equivalent to our definition above
but will at times be more useful in derivations of explicit formulas.
\begin{proposition}\label{prop:conditionLH}\rm
$(M, \omega, H)$ admits a scaling symmetry of degree $\Lambda$ if and only if there exists a primitive, $\lambda$, of $\omega$  ($d\lambda = \omega$), 
satisfying
\begin{equation}\label{eq:conditionLH} 
i_{X_H}\lambda = \Lambda H.
\end{equation}
\end{proposition}
\begin{proof}
One takes $\lambda = i_\D\omega$ so that $d\lambda = \omega$ is equivalent to $L_{\D}\omega=\omega$. Moreover, one then has the relation  $i_{X_H}\lambda=L_{\D}H$ yielding the equivalence of the condition in Eq.~\eqref{eq:conditionLH} to (ii) of Definition \ref{def:dynsim}.
\end{proof}

Now we are ready to ask ourselves the main question: 
%	\begin{center}
given that many symplectic Hamiltonian systems admit a scaling symmetry, how may one use it in order to
\emph{reduce} the system to one defined on a lower-dimensional manifold?
%	\end{center}
Furthermore,
%	\begin{center}
can we do so in such a way that the reduced system has a Hamiltonian (resp. variational) structure?
%	\end{center}
Note that if we can prove such equivalence, then we consider Poincar\'e's dream of an equivalent description
of the physical reality in terms of scale-invariant quantities to be fulfilled.

In the next sections we will address these questions by steps.

\subsection{Scaling symmetries: from symplectic mechanics to contact mechanics}\label{sec:symtocont}

Once a scaling symmetry $\D$ is in hand, the idea is to quotient by its action and obtain a reduced dynamics on a lower-dimensional manifold. 

More precisely, we consider the projections of trajectories of the Hamiltonian system to 
the quotient of $M$ by  $\D$: $M/\D = M/\sim$ where $m\sim m'$ when $m$ and $m'$ 
lie on a common integral curve of $\D$. In general this quotient space need not be a manifold. 
Here we will simply assume that $M/\D$ is a manifold, 
e.g.~that the $\D$'s flow acts freely ($\D$ has no zeroes) and properly, with submersion $M\to M/\D$.

The first step is to describe the additional structure on this reduced space 
which may be used to characterize these projected trajectories. 
%establish conditions under which the space thus obtained is a manifold with some additional structure.
This is the content of the next results.

\begin{proposition}\label{prop:quotient}\rm
Suppose the flow of $\D$ is complete, acting freely and properly on $M$ so that $C:=M/\D$ 
is a smooth manifold with submersion $\pi:M\to C$. Set $\lambda:= i_\D\omega$. Then
\begin{itemize}
    \item[(i)] $C$ is a contact manifold with contact distribution $\mathscr{D} = \pi_*(\ker\lambda)$,
    
    \item[(ii)] on $\mathscr{D}$ there is a conformal symplectic structure,
    
    \item[(iii)] there is a local (exact) symplectomorphism from $(M,d\lambda)$ 
    to the symplectification $(\tilde C,d\tilde \alpha)$ of $C$.
\end{itemize}
\end{proposition}

\begin{proof}
(i): see e.g.~\cite{geiges2001brief}, pg.~36. 
Since there are no zeroes of $\D$, we have local coordinates on $C$ by taking a local transverse slice, 
$\Sigma$, to $\D$. Then with $\lambda = i_\D\omega, L_\D\omega = \omega$ we have $\lambda\wedge (d\lambda)^n = i_\D \omega^n$ 
is non-degenerate on $\Sigma$ determining a contact distribution $\ker\lambda|_\Sigma$ on $\Sigma$, i.e.~$\pi_*\ker\lambda$ on $C$.

(ii): this is a familiar property of contact manifolds. Namely when a contact distribution, 
$\mathscr{D}$, is given by the kernel of a 1-form $\eta$ then $d\eta|_{\mathscr{D}}$ is a symplectic 
form on $\mathscr{D}$. Changing $\eta$ to $f\eta$, with $f\ne 0$, modifies $d(f\eta)|_{\mathscr{D}}$ to $fd\eta|_{\mathscr{D}}$,
and therefore there is a conformal symplectic structure on $\mathscr{D}$.
In our situation, this conformal symplectic structure is represented by $[\omega(\tilde u_1, \tilde u_2)]$ where $\pi_*\tilde u_j = u_j\in\mathscr{D}$. 

(iii): recall (see Definition~\ref{def:symplectification}) that $\tilde C = Ann(\mathscr{D})\backslash C \subset T^*C$, 
the symplectification of $C$, is
an $\R^\times$ principal bundle over $C$ with exact symplectic structure, $d\tilde\alpha$, 
given by restriction of the standard symplectic structure on $T^*C$. 
Consider $\fe : M\to \tilde C$ defined by 
    $$\left(\fe(m), \pi_*v\right) := \lambda_m(v)$$
for all $v\in T_mM$ 
and
with $(\cdot, \cdot )$ the natural pairing of a vector space and its dual. 
One verifies that 
$$
\left(\left(\fe^*\tilde\alpha\right)_m, v\right)
\overset{\text{def of}\,\,\fe^*}{=} \left(\tilde \alpha_{\fe(m)}, \fe_* v\right)
\overset{\text{def of}\,\,\tilde\alpha}{=}\left(\fe(m), (\tilde\pi_C)_*\fe_* v\right)
=\left( \fe(m), \pi_* v\right)
\overset{\text{def of}\,\,\fe}{=}\lambda_m(v)\,\quad \forall v\in T_m M\,,
$$
that is, $\fe^*\tilde\alpha = \lambda$, and so $\fe$ is a (local) exact symplectomorphism.
%\st{and so a local diffeomorphism}.
\end{proof}

By item (ii) in Proposition~\ref{prop:quotient}
there is sense in speaking of orthogonal complements with respect to the conformal symplectic structure in $\mathscr{D}$. We will denote these complements by $^\perp$. 
As for the projected orbits in $C$, we have

\begin{proposition}\label{prop:linefield} \rm
The line field $\text{span}(X_H)$ on $M$ determines a line field $\ell := \pi_* \text{span}(X_H)$ on $C$. 
Integral curves of $\ell$ are projections of orbits of the Hamiltonian system to $C$.  Moreover, letting $\Sigma_0 := \pi(H=0)$, we have 
\[\ell|_{\Sigma_0} = (T\Sigma_0\cap\mathscr{D})^\perp\subset T\Sigma_0 \cap \mathscr{D}.\]
\end{proposition}

\begin{proof}
$\ell$ is well defined from the scale invariance of $\text{span}(X_H)$: $L_{\D}X_H = (\Lambda - 1)X_H$. Since the conformal symplectic structure of Proposition \ref{prop:quotient} (ii) is given by $[\omega(\tilde u_1, \tilde u_2)]$ with $\pi_*\tilde u_j = u_j\in \mathscr{D}$, we will show that $\omega(\tilde u, X_H) = 0$ for any $\pi_*\tilde u\in T\Sigma_0\cap\mathscr{D}$.
First, note that $\D$ is tangent to $H=0$ and by $i_{X_H}i_{\D}\omega = \Lambda H$, that $X_H|_{H=0}\in \ker\lambda$. 
Hence $\ell|_{\Sigma_0}\subset\mathscr{D}$. Let $u\in T\Sigma_0\cap \mathscr{D}$. 
Any lift, $\pi_*\tilde u = u$, of $u$ lies in $T(H\inv(0)) \cap \ker\lambda$ 
and so $\omega(\tilde u, X_H) \sim dH(\tilde u) = 0$, since $\tilde u$ is tangent to $H=0$.
\end{proof}

Note that for degree zero scaling symmetries, $\Lambda = 0$, then $\D$ is tangent to each energy level, and the last proposition applies to arbitrary projections of energy levels to determine the line field $\ell$ on $C$. From here on we will consider the case $\Lambda\ne 0$. As for describing this line field more explicitly in the general case, we first define:
	\begin{definition}\label{def:scalingf}\rm
	A {\em scaling function} for $\D$ is a function $\rho:M\to \R$ such that $L_{\D}\rho = \rho$.
	\end{definition}
And we may now establish the equivalence
	\begin{proposition}\label{prop:equivalence}\rm
	Given $\D$, the following are equivalent:
	\begin{itemize}
	\item[(i)] the existence of a {\em global scaling function}, $\rho:M\to\R^\times$, for $\D$
	\item[(ii)] the existence of a global contact 1-form ($\ker\eta = \mathscr{D}$) on $C$,
	\item[(iii)] an embedding $\iota:C\to M$ as a slice of the $\D$-action: $\pi\circ \iota = id$.
	\end{itemize}
	\end{proposition}

	\begin{proof}
	The equivalence of (i) and (ii) is through the relation $\pi^*\eta = \lambda/\rho$. 
	(ii) and (iii) is through $\iota^*\lambda = \eta$. 
	(iii) and (i) is through $\iota(x) = m$ s.t. $\rho(m) = 1$.
	\end{proof}
	
	\begin{example}[Kepler scaling functions]\label{ex:Kep_scal}\rm    
	Consider the Kepler problem from Example~\ref{ex:Kepler}. 
	Here there are a number of natural choices for scaling functions available, for example any of:
    \[ |q|^{1/2} ,~~ \frac{1}{|p|} ,~~ p\cdot iq, ~~ p\cdot q \]
    are scaling functions of $\D_K = 2q\cdot\del_q - p\cdot \del_p$ (a degree $-2$ scaling symmetry of the Kepler Hamiltonian, 
    $H_K = \frac{|p|^2}{2} - \frac{1}{|q|}$). Observe that $p\cdot iq$, 
    is the angular momentum and $p\cdot q$ is (half) the rate of change of the moment of inertia, $q\cdot q$. The quotient, $C = M/\D$, 
    may be identified with $S^1\times\C$.
	\end{example}

%	\begin{remark}
The items of Proposition~\ref{prop:equivalence} may fail to exist globally. 
Indeed they are equivalent to the symplectification of $C$ being a trivial $\R^\times$ bundle, with connected components, 
$C_\pm$, symplectomorphic to $M$. 
In what follows, we will simply assume that there is a global -- nowhere vanishing on $M$ -- scaling function, 
or if the reader prefers that we are working locally over the set $\rho\ne 0$, where the scaling function is non-vanishing. 
Essentially, the {\em choice} of such a $\rho$ leads to explicit coordinate expressions, 
see Eqs.~\eqref{eq:contCoords}--\eqref{eq:LCHeq2},
as well as the following relation to contact Hamiltonian flows, generalizing
Arnold's description of contact Hamiltonian vector fields % are in 1:1 correspondence
%with symplectic Hamiltonian vector fields on the contactification which are homogeneous of degree 1 with respect to 
%a Liouville vector field
(see Proposition~\ref{prop:arnoldSHVSCH} and Corollary~\ref{cor:D1} below).
%	\end{remark}

	\begin{theorem}[Contact reduction by scaling symmetries: general case]\label{th:CHS}\rm
Let $\rho : M\to \R^\times$ be a global scaling function, with corresponding contact form $\pi^*\eta = \lambda/\rho$ on $C$. With $\Sigma_0:=\pi(H=0)$,
	\begin{itemize}
	\item[(i)] the contact Hamiltonian $\pi^*\mathscr{H}_0 := -H/\rho^\Lambda$ has contact Hamiltonian vector field spanning $\ell$ on $\Sigma_0$,
	\item[(ii)] the contact Hamiltonian $\mathscr{H} :=-|\mathscr{H}_0|^{1/\Lambda}$ has contact Hamiltonian vector field spanning $\ell$ on $C\backslash \Sigma_0$.
	\end{itemize}
	We call this reduction a \emph{contact reduction by scaling symmetries}, or simply a \emph{contact reduction}.
	\end{theorem}

	\begin{proof}
    Recall from Proposition~\ref{prop:arnoldSHVSCH}
	that Hamiltonian flows of degree one functions $F$ on $M$ 
	($L_{\D}F = F$) commute with the scaling action of $\D$ 
	and so induce  contact Hamiltonian flows associated with $\eta$. 
	Item (i) follows by considering the degree one function $H' :=\rho^{1-\Lambda}H$
	and 
	item (ii) with the degree one function $|H|^{1/\Lambda}$. 
	Note that the Hamiltonian vector fields, $X_H$ and $X_{H'}$, are proportional over $H = 0$, while those of $X_H$ and $X_{|H|^{1/\Lambda}}$ are proportional over $H\ne 0$.
	\end{proof}

Before proceeding, some further comments are in order:

	\begin{remark}\rm
	Case (i) in Theorem~\ref{th:CHS} is exactly the case considered in~\cite{sloan2018dynamical}.
	In many relevant physical examples the scaling function $\rho$ can be so chosen so that
	it exists globally on $M$ and is physically interpreted as the overall scale of the system, thus being
	an irrelevant (unmeasurable) degree of freedom for intrinsic observers. Therefore 
	the system can be reduced by eliminating this overall scale and considering only the relational (shape) degrees of freedom 
	(see also~\cite{sloan2021scale,gryb2021scale}).
	 
	Case (ii) in Theorem~\ref{th:CHS}
	%The line field on $C\backslash \bar H_o = \pi(H\ne 0)$ 
	is actually somewhat well-known. 
	Namely, there is always over $C\backslash {\Sigma_0}$ a contact 1-form and an embedding into $M$ 	
	associated to the scaling function $\rho = |H|^{1/\Lambda}$ 
	(the associated embedding being given by restricting to the energy surfaces $H = \pm 1$). 
	Then $\eta = \lambda|_{H = \pm 1}$ and $\ell$ is defined by $i_{\ell}d\eta = 0$ 
	(i.e.~the dynamics may be parametrized as the Reeb flow of $\lambda|_{H = \pm 1}$), see~\cite{MR2397738,bravetti2020invariant}. 
	\end{remark}

	\begin{remark}\label{remark:restr}\rm
	In practice, the reduction may be carried out by embedding $C$ as a hypersurface 
	$\Sigma := \{\rho = 1\}$ transverse to $\D$. 
	Then $\eta  = \lambda|_{\Sigma}$ and $\mathscr{H}_0 = -H|_{\Sigma}$. 
	An equivalent way to state Theorem~\ref{th:CHS} (without cases), 
	is using Definition~\ref{def:LambdaHvf} below. 
	Then the {\em $\Lambda$-Hamiltonian vector field} of $\H_0$ spans $\ell$ on $C$ (see Corollary~\ref{cor:coords} below).
	\end{remark}

By Theorem~\ref{th:CHS}, one may treat the general reduction of the symplectic Hamiltonian system to a {\em contact system} 
locally as two subcases: 
with the degree one $|H|^{1/\Lambda}$ over $C\backslash \Sigma_0$, or $\rho^{1-\Lambda}H$ over $\Sigma_0$, for some scaling function $\rho$.
It is important to stress that for $\Lambda\neq 1$ 
none of these two cases can be extended as contact vector fields (whose flows preserve $\mathscr{D}$) to the other region. 
While this is obvious for 
$|H|^{1/\Lambda}$ over $\Sigma_0$, it is not so immediate for $\rho^{1-\Lambda}H$ over $C\backslash \Sigma_0$, and this will be explained in
more detail shortly (see Corollary~\ref{cor:coords}).
Clearly there is a very important special situation in which the two subcases coincide, i.e.~when $\Lambda=1$.
In this situation, the two contact Hamiltonian flows in Theorem~\ref{th:CHS} coincide and we obtain the following reduction: 
	\begin{corollary}[Contact reduction by scaling symmetries: degree one case]\label{cor:D1}\rm
	Let $(M, \omega, H)$ be a Hamiltonian system admitting a scaling symmetry $\D$ of degree one. 
	Then:
	\begin{itemize}
	\item[i)] $X_H$ defines a vector field $X = \pi_*X_H$ on $C$,
	\item[ii)] for a scaling function $\rho$, and $\pi^*\eta = \lambda/\rho$, $X$ is the contact Hamiltonian vector field of 
	$\pi^*\mathscr{H} = - H/\rho$.
	\end{itemize}
	\end{corollary}
By virtue of the above Corollary, in the case $\Lambda=1$ one obtains a `full' contact reduction, 
meaning that the reduced contact system is completely determined
by the original symplectic one (this is precisely the content of Proposition~\ref{prop:arnoldSHVSCH}, 
which we have recovered as a particular case).
	Note also that in this case
	%the $\Lambda=1$ case, 
	$H$ itself may be chosen as a scaling function giving the reduced dynamics as a 
	Reeb flow of $\pi^*\eta = \lambda/ H$ on $C\backslash \Sigma_0$.
	
In the general case, the line field we have been considering may in fact be described more explicitly with the use of a scaling function,
as detailed in the following
	\begin{corollary}[Dependence on the scaling function]\label{cor:coords}\rm
	Let $(M, \omega, H)$ be a Hamiltonian system admitting a scaling symmetry $\D$ of degree $\Lambda$ and $\rho$ a scaling function for $\D$.
	Then the line field $\ell$ on $M/\D$ is spanned by
	\begin{equation}\label{eq:lineSpan}
	    X := X_{\mathscr{H}_0} + (1-\Lambda)\mathscr{H}_0 \mathscr{R}\,,
	\end{equation}
	where $X_{\mathscr{H}_0}$ is the contact Hamiltonian vector field of $\pi^*\mathscr{H}_0 = - H/\rho^\Lambda$ 
	with respect to $\pi^*\eta = \lambda/\rho$ and $\mathscr{R}$ is the Reeb vector field of $\eta$ on $C$. 
	Moreover, for any other choice of scaling function, $\tilde\rho$, we have
	\[ \tilde X = \left(\frac{\tilde\rho}{\rho}\right)^{1 - \Lambda} X.\]
	\end{corollary}
	
	\begin{proof} Given a scaling function, $\rho$, then by Eq.~\eqref{eq:scalingXH}, the vector field $\rho^{1-\Lambda}X_H$
	is scale invariant and projects to a vector field, 
	\[X = \pi_* \rho^{1-\Lambda} X_H,\]
	on $C$ spanning $\ell$. Set $H' := \rho^{1-\Lambda}H$, a degree one Hamiltonian on $M$ 
	with $\pi_*X_{H'} = X_{\mathscr{H}_0}$ the contact Hamiltonian vector field of 
	$\pi^*\mathscr{H}_0 =  - H/\rho^\Lambda$ with respect to $\pi^*\eta = \lambda/\rho$. 
	Then
	\[ dH' = (1 - \Lambda) \rho^{-\Lambda}H~d\rho + \rho^{1-\Lambda}~dH\]
	so that
	\be\label{eq:XHVsXH'} 
	X_{H'} = (1 - \Lambda) \rho^{-\Lambda}H~X_\rho + \rho^{1-\Lambda}~X_H\,.
	\ee
	Note that by Corollary~\ref{cor:D1}, the Hamiltonian vector field $X_\rho$ of $\rho$ projects 
	to the Reeb vector field $\mathscr{R} = \pi_*X_\rho$ of $\eta$ on $C$. 
	Applying $\pi_*$ to the last equation we have Eq.~\eqref{eq:lineSpan}.
	
	For the last relation, note that for two scaling functions, $\rho$ and $\tilde\rho$, we have that $\tilde\rho/\rho$ is scale invariant, 
	so there is a function $\sigma$ on $C$ through $\pi^*\sigma = \tilde\rho/\rho$. Then
	\[ \tilde X = \pi_* \tilde\rho^{1-\Lambda}X_H = \sigma^{1-\Lambda}X.\] %\pi_* \left(\frac{\tilde\rho}{\rho}\right)^{1-\Lambda} \rho^{1-\Lambda}X_H

		%\item[(iii)] For any other choice of scaling function, $\tilde\rho$, the contact Hamiltonian vector field, $X_{\mathscr{\tilde H}_0}$, of $\pi^*\mathscr{\tilde H}_0 = H/\tilde\rho^\Lambda$ with respect to the contact form $\pi^*\tilde\eta = \lambda/\tilde\rho$ satisfies:
	%\[X_{\mathscr{\tilde H}_0} = \left(\frac{\rho}{\tilde\rho}\right)^{\Lambda - 1} X_{\mathscr{H}_0}.\]
	
	%	For item (iii), we have that $L_\D \left(\frac{\rho}{\tilde\rho}\right) = 0$ so 
	%that there is a function, $\sigma$, on $C$ with $\pi^*\sigma = \frac{\rho}{\tilde\rho}$. 
	%The relation on the contact Hamiltonian vector fields follows from 
	%$\pi^*\tilde\eta = \frac{\rho}{\tilde\rho} \lambda/\rho = \pi^*\sigma\eta$ 
	%so that $\sigma \eta = \tilde\eta$ and similarly $\sigma^\Lambda\mathscr{H}_0 =  \mathscr{\tilde H}_0$.
	
	\end{proof}
	
	\begin{remark}\rm\label{rem:coordinates}
	In Darboux coordinates $\eta = dS - p_a dq^a$ on $C$, Eq.~\eqref{eq:lineSpan} reads
	    \begin{equation}\label{eq:contCoords}
	         (q^a)' = \del_{p_a}\mathscr{H}_0 , \qquad
         (p_a)' =  - \del_{q^a}\mathscr{H}_0 {-} p_a\del_{S}\mathscr{H}_0 , \qquad
         S' =  p_a\del_{p_a}\mathscr{H}_0-\Lambda \mathscr{H}_0\,. 
	    \end{equation}
	These equations may also be derived directly from Hamilton's equations for $H$ by following 
	Remark~\ref{remark:restr} or the comment after Proposition~\ref{prop:arnoldSHVSCH}. 
	Namely, given a scaling function $\rho$, we take $\Sigma = \{ \rho = 1\}$. 
	Then a system of contact coordinates, $\eta = dS - p_a dq^a$ on $C$, 
	may be extended to symplectic coordinates 
	$P_0 = \rho , Q_0 = S, P = - \rho p, Q = q$ on $M$ with $\lambda = \rho\eta = P_0dQ_0 + P\cdot dQ$. 
	Then, we have that
	\[ H(Q_0, Q, P_0, P) = -P_0^\Lambda\mathscr{H}_0(Q_0, Q, -P/P_0)\]
	since $H/\rho^\Lambda = - \mathscr{H}_0$.
%	$P_0 = S, Q_0 = \rho, P_a = \rho p_a, Q = q^a$ on $M$, 
%	with $\lambda = P_0 dQ_0 + P_a dQ^a =\textcolor{red}{-\rho\eta+d(\rho S)}$. 
%	\textcolor{red}{Then we have that
%	\[H(Q_0, P_0, Q^a, P_a) = Q_0^\Lambda \mathscr{H}_0(P_0, Q^a, P_a/Q_0)=\mathscr{H}_0(S,q^a,p_a)\]
	By chain rule, we may write Hamilton's equations of motion as:
	    \begin{eqnarray}
        \rho^{1 - \Lambda}\dot S &=&  p_a\del_{p_a}\mathscr{H}_0-\Lambda \mathscr{H}_0  , \qquad 
        \rho^{1 -\Lambda}\dot \rho =  \rho \del_S\mathscr{H}_0,\label{eq:LCHeq1}\\
        \rho^{1 -\Lambda}\dot q^a &=& \del_{p_a}\mathscr{H}_0 , \qquad\qquad\qquad\quad
        \rho^{1 -\Lambda}\dot p_a =   - \del_{q^a}\mathscr{H}_0 {-} p_a\del_{S}\mathscr{H}_0 \,, \label{eq:LCHeq2}
    \end{eqnarray}
    which project onto $\Sigma = \{ \rho = 1\}$ as Eqs.~\eqref{eq:contCoords}. 
    It is also evident from Eqs.~\eqref{eq:LCHeq1},~\eqref{eq:LCHeq2} of this last computation that, 
    under the reparametrization $\rho^{1-\Lambda}d\tau = dt$, 
    the scale-reduced equations of motion~\eqref{eq:contCoords} contain a `blow up',  
    of Hamilton's equations of motion to $\rho = 0$, 
    as one would expect, compare with~\cite{MontgomeryBlowup} for $n$-body problems.
    \end{remark}

    \begin{remark}\rm
    The flow of the vector field $X$ in Eq.~\eqref{eq:lineSpan}, 
    induced by a scaling function $\rho$, does not in general preserve $\mathscr{D}$. 
    However restricted to the invariant set $\Sigma_0 = \{\mathscr{H}_0 = 0\}$ it does, 
    being a contact Hamiltonian vector field (case~(i) of Theorem~\ref{th:CHS}). 
    In general, one may rescale $X$ away from this set so that its flow does preserve $\mathscr{D}$. 
    Namely, for the scaling function $\tilde\rho = |H|^{1/\Lambda}$, 
    we have $\mathscr{\tilde H} = - 1$, and $\tilde X = \Lambda \mathscr{\tilde R}$ is proportional to the 
    Reeb vector field of $\tilde\eta$ (case~(ii) of Theorem~\ref{th:CHS}) 
    so that $|\mathscr{H}_0|^{1 - 1/\Lambda} X = \tilde X$ preserves $\mathscr{D}$.
    \end{remark}

	Another interesting aspect about item (i) in Theorem~\ref{th:CHS} is the fact that it provides a direct connection with the physics
	of the problem. Indeed, the function $\rho$ is usually connected with a global scale within the physical description 
	(hence the `scaling function' name), 
	and therefore
	by
	using it in order to reduce the dynamics we automatically obtain a description of the same physical problem in terms of 
	scale-invariant
	functions only.
	
	While it may seem that the `ambiguity' in Corollary~\ref{cor:coords} due to one's choice of scaling function is 
	a defect of the general contact reduction, we posit instead that this `freedom' to choose a scaling function to describe 
	the projected trajectories is in fact an asset of the theory, allowing one to highlight certain aspects of the original 
	systems dynamics by various choices of scaling function.
	We illustrate the reduction procedure with the Kepler problem in the following example.

	\begin{example}[Contact-reduced Kepler]\label{ex:scaledK}\rm
    Consider the Kepler problem from Example~\ref{ex:Kepler}. %To have a better intuition of the physical aspects of the reduction to scale-invariant quantities, 
    To illustrate the contact reduction process, one may first consider polar coordinates:
    \[ q = re^{i\theta} ,~~ p = \left(p_r + i \frac{p_\theta}{r}\right) e^{i\theta},\]
    with scaling functions (Example~\ref{ex:Kep_scal}) that we denote:
    \[ \rho := |q|^{1/2} = r^{1/2} ,~~ J:= p\cdot q = rp_r ,~~ G:= p\cdot iq = p_\theta.\]
    In the coordinates $(\rho, \theta, J, G)$ on $M = T^*(\C\backslash 0)$, we have 
    \begin{equation}\label{eq:bSymplectic}
        \omega = \frac{2 dJ\wedge d\rho}{\rho} + dG\wedge d\theta,
    \end{equation}
    and
    \begin{equation}\label{eq:KepContRed}
        \D_K = \rho\del_\rho + J\del_J + G\del_G ,\qquad \lambda = i_{\D_K}\omega = G~d\theta - 2J~d\log\frac{J}{\rho} ,\qquad 
        H_K = \frac{J^2 + G^2}{2\rho^4} - \frac{1}{\rho^2}\,.
    \end{equation}
    The functions $J/\rho$, $G/\rho$ and $\theta$ are $\D_K$ invariant, determining coordinates $(\tilde J, \tilde G, \theta)$ on the quotient $M/\D_K$ by:
    \[ \pi^*\tilde J = J/\rho , \qquad \pi^*\tilde G = G/\rho, \qquad\pi^*\theta = \theta\]
    where $\pi: M\to M/\D_K = C$. 
    According to Theorem~\ref{th:CHS}, the scale-reduced orbits may be described on $C$ using:
    \begin{equation}\label{eq:KepRedCoords}
        \eta = \tilde G~d\theta - 2d\tilde J ,\qquad \mathscr{H}_0 = 1 - \frac{\tilde J^2 + \tilde G^2}{2}
    \end{equation}
    with $\pi^*\eta = \lambda/\rho$ and $\pi^*\mathscr{H}_0 = - \rho^2 H_K$. 
    More explicitly, by Corollary~\ref{cor:coords} or Remark~\ref{rem:coordinates} with $\Lambda = -2$, 
    the scale-reduced equations of motion may be written:
    \begin{equation}\label{eq:KepRedEoM}
        \tilde J' = \frac{\tilde G^2}{2} - \mathscr{H}_0 ,\qquad \tilde G' = - \frac{\tilde G\tilde J}{2} ,\qquad\theta' = \tilde G.
    \end{equation}
   
	\end{example}

\begin{remark}\rm
    The scaling symmetry, $\D_K$, of the Kepler problem and its corresponding scale-invariant functions, 
    e.g.~$\tilde J, \tilde G, \mathscr{H}_0$, have long been known and used in celestial mechanics, 
    in particular their analogues for $n$-body problems. 
    See for example Section 2.3 of Chenciner's~\cite{Chenciner1997alinfini}, where relations of the scale reduction to McGehee blow-up are explained. 
    On the other hand,  the contact structure associated to these contact reductions has been emphasized only more recently~\cite{sloan2018dynamical}.
    Observe that the contact reduction here is closely related to the structure of `$b$-manifolds' in~\cite{MirandaetalSingSymp}. 
    In particular compare eq.~\eqref{eq:bSymplectic} here to the Darboux normal form (Theorem 2 of~\cite{MirandaetalSingSymp}) 
    of a `$b$-symplectic form' and Example 8.2 of~\cite{miranda2018geometry} to our construction here of a scale-reduced contact form.
\end{remark}

An important property of the reduction by Theorem~\ref{th:CHS} is its relation to certain blow-ups as mentioned at the end of 
Remark~\ref{rem:coordinates}. We illustrate this relation with the scale-reduced Kepler problem.

%Another illuminating example about the importance of the reduction by Theorem \ref{th:CHS} is the following.

    \begin{example}[Kepler blow-up]\label{ex:Kblowup}\rm We continue from the previous example 
    with the scale-reduced Kepler problem using the global scaling function $\rho = |q|^{1/2}$. 
    According to Eqs.~\eqref{eq:LCHeq1} and~\eqref{eq:LCHeq2} with $\Lambda = -2$, the scale-reduced equations 
    of motion on $\pi^*\mathscr{H}_0 = -\rho^2 H_K = 0$, represent a blow-up of the dynamics at $\rho = 0$ (collision).
    
    One may apply Proposition~\ref{prop:linefield} to determine the scale-reduced orbits on $\mathscr{H}_0 = 0$. 
    From Eq.~\eqref{eq:KepRedCoords}, we have:
    \[\mathscr{D} = \ker\eta = \text{span} \{ 2\del_\theta + \tilde G \del_{\tilde J} , \del_{\tilde G}\}\]
    while $\Sigma_0 = \{\mathscr{H}_0 = 0\}$ is a torus, parametrized by the angle $\theta\in S^1$ and the circle $\tilde J^2 + \tilde G^2 = 2$. So:
    \[ T\Sigma_0 = \text{span}\{ \del_\theta, \tilde G \del_{\tilde J} - \tilde J \del_{\tilde G}\}.\]
    The collision orbits of the Kepler problem are the homothetic motions. 
    They tend in forwards or backwards time to the fixed points of eqs.~\eqref{eq:KepRedEoM}, $\tilde J = \pm \sqrt{2}, \tilde G = 0$.
    Observe that $\mathscr{D}$ is tangent to $\Sigma_0$ exactly at these fixed points. 
    Away from these points, by Proposition~\ref{prop:linefield}, the projected Kepler orbits on $\Sigma_0$ are integral curves of:
    \[ \ell|_{\Sigma_0} = \mathscr{D}\cap T\Sigma_0 = \text{span}\{ 2\del_\theta + \tilde G \del_{\tilde J} - \tilde J \del_{\tilde G}\}.\]
    Taking the ($\D_K$-invariant) angle $\fe$ by $p = |p|e^{i\fe}$, we have the line field:
    \[\ell|_{\Sigma_0} = \text{span}\{  \del_\fe + 2\del_\theta \}.\]

    The scale-reduced dynamics on this torus, $\Sigma_0$, are a blow-up at collision $\rho = |q|^{1/2} = 0$ 
    of the equations of motion for the Kepler problem. We recover the well-known `elastic bounce' regularization 
    and collision torus (see \S 1.3 of \cite{devaneySingCM}), although our path here has been rather different.
    \end{example}
    
    \begin{remark}\rm
    In Appendix~\ref{app:Kepler} we outline some more cases of contact reductions by scaling symmetries in celestial mechanics.
    \end{remark}

%special case of our thm above). \qed\\

We close this section with an interesting comment about the relationship between Corollary~\ref{cor:D1} and
the contact reduction in~\cite{valcazar2018contact}, which in turn gives yet another proof of Arnold's Proposition~\ref{prop:arnoldSHVSCH}
in terms of a contact version of symplectic reduction. %Noether's theorem.
	\begin{remark}\label{rmk:deLeonreduction}\rm
Corollary~\ref{cor:D1} may be alternatively deduced as a special case of the contact reduction in~\cite{valcazar2018contact}.
The relevant case is for an $\R$ action on a contact manifold $(N, \eta)$ by {\em exact} 
contactomorphisms:
$L_X\eta = 0$, where $X\in\mathfrak{X}(N)$ generates the $\R$ action 
and where $\mathscr{H}$  is an $\R$-{\em invariant} contact Hamiltonian on $N$. 
Let $J_0 := (i_X\eta)\inv(0)$. 
Then $X, X_{\mathscr{H}}$ are tangent to $J_0$ and the quotient $P_0 := J_0/\R$ 
is a contact manifold with reduced contact dynamics defined through: 
\begin{center}
\begin{tikzcd}
& N \\
J_0 \arrow[ru, "\iota_0"] \arrow[dr, "\pi_0"] & \\
  & P_0\\
\end{tikzcd}
\[\qquad \pi_{0*}X_{\mathscr{H}} = X_{\mathscr{\bar H}} ~~{\text{ for }} 
~~\pi_0^*\mathscr{\bar H} = \iota_0^*\mathscr{H}, 
~~\pi_0^*\bar\eta = \iota_0^*\eta.\]
\end{center}
%\[\pi_{o*}X_{\mathscr{H}} = X_{\mathscr{\bar H}} ~~{\text{ for }} 
%~~\pi_o^*\mathscr{\bar H} = \iota_o^*\mathscr{H}, 
%~~\pi_o^*\bar\alpha = \iota_o^*\alpha.\]

{Applying this construction, we obtain}
	\begin{proposition}\label{prop:contactreduction}\rm
	Consider a Hamiltonian system $(M, \omega, H)$ admitting $\D$ as a degree one scaling symmetry. For $\lambda = i_\D\omega$ and a global scaling function $\rho : M\to \R^\times$, let $N := \R\times M$ with 
	$\eta := e^{-t}(\lambda + d\rho) - dt$ where $t$ is the coordinate on $\R$.
	Then 
	\begin{itemize}
	\item[(i)]
	$N$ is a contact manifold and $\tilde \D = \del_t + \D$ generates an $\R$-action on $N$ by exact contactomorphisms of $\eta$,	
%	\alescomment{we have a problem with $S=0$. Should we use $\exp(t)$ instead of $S$?}
	\item[(ii)]
	the contact Hamiltonian $e^{-t}H$ is invariant under the flow of $\tilde \D$.
	\item[(iii)] The contact manifolds $C = M/\D$ and $P_0 = J_0/\tilde\D$ are contactomorphic and moreover the reduced contact dynamics 
	(of~\cite{valcazar2018contact}) on $P_0$ corresponds to the contact Hamiltonian dynamics on $C$ (of Theorem~\ref{th:CHS}).
	
	%The 
	%the diagram:  \begin{tikzcd}
	% N  \arrow[d] &  \\
	%M \arrow[d]\arrow[r, "\cong"]  & J_o \arrow[ul]\arrow[d]\\
	%C \arrow[r, "\cong"] & P_o\\
	%\end{tikzcd}
	%commutes. 
	
	%The contact-reduced flow on $P_o$ corresponds to $X_{\mathscr{H}}$ of $\pi^*\mathscr{H} = H/\rho$, with respect to $\pi^*\eta = \lambda/\rho$ on $C$.
	\end{itemize}
	\end{proposition}

\begin{proof} 
{According to Remark~\ref{rmk:deLeonreduction},} 
the corresponding moment map is $i_{\tilde \D}\eta =  e^{-t} \rho - 1$ and so $J_0$ is the graph $\{e^t = \rho\}\subset N$, 
identified with $M$ via $\iota: m\mapsto  (\log \rho(m), m)$. 
Under this identification, we have $\iota_*\D =  \del_t + \D = \tilde \D|_{J_0}$ so that $C = M/\D\cong P_0 = J_0/\tilde \D|_{J_0}$. 
Moreover, $\iota^*\eta = \lambda/\rho$  and $\iota_*(e^{-t} H) = H/\rho$. 
Hence, the contact-reduced flow corresponds to the contact Hamiltonian flow of $\pi^*\mathscr{H} = H/\rho$ with respect to $\pi^*\eta = \lambda/\rho$.
\end{proof}

    \end{remark}

\section{Can one always find a contact reduction?}\label{sec:globalreduction}
Given the above discussion on the general existence of a contact reduction ensuing from Theorem~\ref{th:CHS} and
Corollary~\ref{cor:D1}, 
we are now interested in the following fundamental question:
\begin{center}
    given a symplectic Hamiltonian system $(M,\omega,H)$ admitting a scaling symmetry $\D$ of degree $\Lambda$,
    can one transform $\D$ to $\tilde \D$ such that the new vector field is a scaling symmetry 
    with $\tilde\Lambda=1$?
\end{center}
This will be the content of this section.
We will first prove that $\tilde\D$ always exists locally and find an explicit algorithm to construct it,
%to find $\tilde \D$, which will tell us exactly which conditions 
%that need to be satisfied in order for $\tilde \D$ to exist globally 
and then we will use
an example in order to make clear that such a symmetry cannot be found globally in general on $M$.
This will lead us to consider an extension of $(M,\omega,H)$
to a `lifted system', where
the parameters are considered as dynamical variables (together with their conjugate momenta). 
Rather surprisingly, in this space we can prove the existence of a scaling symmetry of degree one, and therefore
obtain its corresponding contact reduction, even for cases in which the original Hamiltonian system on $M$ had
no (evident) scaling symmetry at all.

\subsection{Locally on $M$: yes}

We begin with the following observation:
for any Hamiltonian system in a neighborhood of $m$ with $X_H(m)\ne 0$, there exist Darboux coordinates such that
$H = p_1$ and $\omega = dp_a\wedge dq^a$. 
Then $\tilde\D = p_a\del_{p_a}$ is a degree one scaling symmetry with $\lambda = p_a dq^a = i_{\tilde\D}\omega$. 
We conclude that locally, away from critical points of $H$, it is always possible to find $\tilde \D$.

%	\begin{center}
%	How can we explicitly construct such coordinates and how far can we extend $\tilde \D$?
%	\end{center}
%These coordinates are not so useful for examples (they are not really possible to find for explicit systems), 
%but they give the general picture for how the reduction works locally.

This observation may be described in the following algorithm to find $\tilde \D$ locally: %for a general Hamiltonian system as follows: 
	\begin{itemize}
	\item[1.] fix {\em some} primitive, $\lambda$, of $\omega$ and seek a closed 1-form $\alpha$ such that 
		\begin{equation}\label{eq:condition1}
		i_{X_H}(\lambda - \alpha) = H\,.
		\end{equation}	
	\item[2.] By the discussion in Section~\ref{subsec:scalingsymm} (see e.g.~Proposition~\ref{prop:conditionLH}),
	the vector field $\tilde \D$ corresponding to $\lambda - \alpha$ via 
	\be\label{eq:recoveringtildeD}
	i_{\tilde\D}\omega=\lambda-\alpha
	\ee
	is the generator of a degree one scaling symmetry.
	\end{itemize}
We conclude that~\eqref{eq:condition1} is the equation to be solved (for $\alpha$) in order to find $\tilde \D$.

Notice that, locally at least, $\alpha=dS$ for some function $S:M\rightarrow\R$ and then by~\eqref{eq:recoveringtildeD} its associated Hamiltonian vector field $X_{S}$	
is such that $\tilde \D=X+X_{S}$, with $X$ the Liouville vector field satisfying $i_{X}\omega=\lambda$.
Locally therefore,
we can make~\eqref{eq:condition1} even more explicit. Indeed, we obtain
	\begin{equation}\label{eq:condition2}
	X_HS = i_{X_H}\lambda - H\,. 
	\end{equation}
In particular, in local Darboux coordinates with $\lambda = p_a dq^a$, this gives
	\be\label{eq:condition2bis}
	\dot S = p_a\dot q^a - H = L\,,
	\ee 
and therefore $S$ is the action.
Remarkably, it is always possible to solve~\eqref{eq:condition2bis} locally 
away from critical points of $H$  by the choice of a transverse slice $\Sigma$ to $X_H$ 
identifying  $M\cong \Sigma\times(-\e, \e)\ni (\s, t)$ and taking $S(\s, t) := \int_0^t L(\fe_\tau(\s))~d\tau$.
However, it is clear that globally there may be obstructions to solving~\eqref{eq:condition2} -- resp.~\eqref{eq:condition1} -- 
in general
(see e.g.~the Kepler case in Section~\ref{sec:Kepler}). Therefore an alternative route should be put forward, 
as we will do in Section~\ref{sec:couplings}.

In order to illustrate the above procedure, we consider now a very simple case in which the scaling symmetry
$\tilde \D$ can be found explicitly:
	\begin{example}[The 2d harmonic oscillator]\label{ex:2DHO}\rm
	We have $H = \frac{|p|^2 + k|q|^2}{2}$, $\lambda = p \cdot dq$ and $\omega = d\lambda$.
	%$H=\frac{p_{r}^{2}}{2}+\frac{p_{\theta}^{2}}{2r^{2}}+k\,r^{2}$, $\lambda = p_a dq^a$ and $\omega=d\lambda$.
	Imposing $i_{X}\omega=\lambda$ we get $X = p\cdot \del_p$
	%$
	%X=p_a\partial_{p_a},
	%$
which is by definition a Liouville vector field but not a scaling symmetry of $H$.

Now we seek a function $S$ such that $\tilde \D=X+X_{S}$ is a scaling symmetry with $\Lambda=1$.
Solving~\eqref{eq:condition2} we get
	\be\label{eq:S2DHO}
	S = \frac{p\cdot q}{2} + c p \cdot iq %S=\frac{1}{2}p_{r}r+cp_{\theta}\,,
	\ee
where $c$ is a constant of integration that we will fix to 0 
(this freedom in fixing $\tilde \D$ comes from the fact that $p\cdot iq$, the angular momentum, is an integral of motion,
cf.~the discussion in Remark~\ref{rmrk:freedom2}). This degree one scaling symmetry with $c = 0$ is just $\tilde\D = 
X + X_S = \frac{p\cdot\del_p + q\cdot \del_q}{2}$.
	\end{example}
Note that in the harmonic oscillator example $S$ (and thus $\tilde\D$) is globally defined on $M$.
In the next section we will see an example in which a global solution to eq.~\eqref{eq:condition1} on $M$ cannot exist.

\subsection{Globally on $M$: no}\label{sec:Kepler}

Let us consider the Kepler problem of Example~\ref{ex:Kepler}.
This admits the `standard scaling symmetry' $\D_K = 2q\cdot\del_{q} - p\cdot\del_{p}$, which has degree $-2$. %i.e.~$L_{\D_}H = - 2H$.
May we apply our algorithm to find a different `hidden scaling symmetry' of degree 1?

By the discussion in the previous section, this question is equivalent to finding a primitive $\lambda'$ of   $\omega_{K}$
with $i_{X_{H_{K}}}\lambda' = H_{K}$. 
The degree one scaling symmetry is then defined by $i_{\tilde\D}\omega_{K} = \lambda'$.
So, let $\lambda = p_a dq^a$ be the canonical 1-form, then all of our options for 
primitives are given by $\lambda - \alpha$ with $\alpha$ a closed 1-form.

It is instructive to consider first the case in which $\alpha = dS$ is exact. Then we look for a solution of~\eqref{eq:condition2}.
As an immediate observation, we see that for the Kepler problem a solution $S$ 
fails to exist globally because there are periodic orbits with positive action. 
Namely, for an elliptic orbit, $\gamma\subset T^*(\R^2\backslash 0)$, 
with period $T$ and (negative) energy $E$, a solution $S$ to~\eqref{eq:condition2} would satisfy:

%are periodic \textcolor{red}{with positive Lagrangian, $S$ cannot be defined over 
%a complete periodic orbit with period $T$, because otherwise}
    \be\label{eq:noway}
 0  \overset{!}{=} \oint_{\gamma}dS  =  S(\s, T) - S(\s, 0) = \int_0^T L(\fe_{\tau}(\s))~d\tau 
 = 3\pi \left( \frac{T}{2\pi}\right)^{1/3} = 
 3\pi \sqrt{\frac{-1}{2E}} > 0\,.
    \ee
%Then $dS$ 

In general, for any closed $\alpha$ as above, consider two periodic elliptic trajectories $\gamma_1, \gamma_2\subset T^*(\R^2\backslash 0)$ with different values of (negative) energies, $E_1\ne E_2$ resp. Such trajectories bound some surface, $\Sigma\subset T^*(\R^2\backslash 0)$, in the phase space.
Then, since $\alpha$ is closed, $d\alpha = 0$, we have by Stokes theorem:
    \begin{equation}\label{eq:Keplerimp1}
    \oint_{\gamma_1}\alpha = \oint_{\gamma_2}\alpha
           % 0=\int_\Sigma d\alpha = \oint_{\gamma_1 - \gamma_2}\alpha=\int_0^{T_1}\left(i_{X_H}\alpha\right)(t)\, dt - \int_0^{T_2}\left(i_{X_H}\alpha\right)(t)\, dt\,,
    \end{equation}
%where in the first equality we have used that $\alpha$ is closed ($d\alpha=0$), in the second we used Stokes theorem, and in the third 
%the fact
%that $\gamma$ is an integral curve of $X_H$ (note $T$ is the period of the trajectory).

As we have fixed $\lambda=p_a dq^a$ as the canonical one-form, eq.~\eqref{eq:condition1} reads
    \begin{equation}\label{eq:Keplerimp2}
        i_{X_H}\alpha=p_a\frac{\partial H}{\partial p_a}-H=L\,.
    \end{equation}
Therefore, since $\gamma_1, \gamma_2$ are integral curves of $X_H$,~\eqref{eq:Keplerimp1} now reads
 \begin{equation}\label{eq:Keplerimp3}
         \int_0^{T_1}L(\gamma_1(t))\, dt = \int_0^{T_2}L(\gamma_2(t))\, dt  %0=\oint_{\gamma}\alpha=\int_0^{T}L(\gamma(t))\, dt\,,
    \end{equation}
where $T_1\ne T_2$ are the periods of $\gamma_1, \gamma_2$. This is impossible, because for the Kepler problem, the total action around an elliptic orbit is a non-constant function of its energy (and period). %$L(\gamma(t))>0$.
We conclude that such a $\alpha$ cannot exist in the Kepler case.

In Example~\ref{ex:2DHO} for the 2d harmonic oscillator we have used precisely this
route to find $\alpha=dS$. 
Note that in that case we have that the condition~\eqref{eq:Keplerimp3} is identically satisfied (both sides always being zero).

\subsection{Globally on $\hat M$: yes. The evolution of couplings}\label{sec:couplings}
Now let us consider a symplectic Hamiltonian system, $(M,\omega,H_a)$, 
with the Hamiltonian depending on some parameters
$a\in\R^k$ (e.g.~the parameters might be masses, gravitational constants, etc.). 
We will define its extension to a `lifted system', where
the parameters are considered as dynamical variables (together with their conjugated momenta) 
and show that for such system a global (degree one) contact reduction
always exists.
	\begin{definition}\label{def:extsys}\rm
	Given a (parametrized) symplectic Hamiltonian system $(M,\omega,H_a)$,
	the associated \emph{lifted system} is the symplectic Hamiltonian system
	$(\hat M,\hat \omega,\hat H)$
	with 
	$$\hat M := M \times \R^{2k}\ni (m, a, b), \qquad
	\hat\omega := \omega + da_i\wedge db^i,
	\qquad \text{and} \qquad  
	\hat H: \hat M\to \R, (m, a, b)\mapsto H_a(m).$$
%	\textcolor{red}{Further, we call $\hat M$ the \emph{space of equivalent Hamiltonian theories}.}
	\end{definition} %remark: cotangent bundle case
Clearly, as a consequence of the definition, we have
	\be\label{eq:hatXH}
	X_{\hat H} = X_{H_a} + \del_{a} H_a \cdot\del_{b}
	\ee
from which we see that $\dot a=0$ (the parameters $a$ 
are first integrals of the lifted system) and thus, after a choice of $a$ has
been made once and for all, the lifted dynamics
always projects down to the original one.

In most cases, one may simply redefine the parameters so that $H_a$ is degree one in $a$, 
which includes for example the lift of {\em any} Hamiltonian system by taking $H_a := aH$. 
Similarly one may consider a system of {\em coupled} Hamiltonians:
\[H_a = a_1H_1 + ... + a_k H_k\]
with parameters $a_j\in\R$ the {\em couplings}.

With the above definitions we have the following important result:
	\begin{proposition}[Lifted degree one scaling symmetries]\label{thm:globalreduction}\rm %Global contact reduction on the space of equivalent theories
	Consider a Hamiltonian system, $(M, \omega, H_a)$ depending on parameters, $a\in \R^k$, of the form:
	\[H_a = a_1H_1 + ... + a_k H_k\]
	with $H_j:M\to\R$, and admitting a Liouville vector field $\D$ such that $L_{\D} H_j = \Lambda_j H_j$. Then 
		$${\hat \D}:=\D+\sum_{i=1}^k{(1-\Lambda_i) a_{i}\partial_{a_{i}}}+\sum_{i=1}^k\Lambda_i b^{i}\partial_{b^{i}}$$
		is a degree one scaling symmetry for the associated lifted system, $(\hat M, \hat\omega, \hat H)$, that is:
	${L_{\hat \D}\hat \omega=\hat \omega}  \text{ and }  {L_{\hat \D}\hat H=\hat H}.$
	%that is, {$\hat \D$ scaling symmetry of degree 1} for ${(\hat M,\hat \omega, \hat H)}$.
	%Let $(M, \omega,H_a)$ be a Hamiltonian system and ${\D}$ a Liouville vector field such that
	%\be\label{eq:conditionglobalreduction1}
	%H_a=\sum_{i=1}^na_iH_i\,, \qquad \text{and} \qquad L_{\D}H_i=\Lambda_iH_i\,,\qquad \forall i=1,\dots, n\,.
	%\ee
	\end{proposition}

    \begin{proof}%[Proof of Theorem~\ref{thm:globalreduction}]
By a direct calculation we have that 
$L_{\hat \D}\hat H=\hat H$ and 
$i_{\hat \D}\hat \omega=\lambda + \sum_{i=1}^ka_i db^i-d\left(\sum_{i=1}^k\Lambda_ia_ib^i\right)$,
where $\lambda = i_{\D}\omega$, and therefore $\hat \D$ is a scaling symmetry of degree one.
	\end{proof}

\begin{remark}\label{remark:coupledHamv2}\rm
   Note that the same construction above, upon replacing $a_i$ with $a_i^{1-\Lambda_i}$ in the definition of the Hamiltonian $H_a$, 
   results in a degree one scaling symmetry 
$${\hat \D} = \D + \sum_{i=1}^k a_i \partial_{a_i}.$$
This is often more useful in directly finding the scale-reduced contact Hamiltonian. 
Furthermore, when the terms $1-\Lambda_i$ are not integers, one may avoid complex values by considering various cases of sign choices in the sum
\begin{equation}\label{eq:coupledHamv2}
    H_a = |a_1|^{1-\Lambda_1}H_1 \pm ... \pm |a_k|^{1-\Lambda_k} H_k\,.
\end{equation}
\end{remark}
	
According to Section~\ref{sec:contactreduction} (Corollary~\ref{cor:D1}), we may project the Hamiltonian dynamics of $\hat H$ 
onto a contact Hamiltonian mechanics of the vector field $\hat X$ on $\hat C = \hat M/ \hat \D$. 
Our main result here describes how the dynamics of each $H_a$ is related to this reduced space. 

Let us observe first that the scaling symmetry $\hat\D$, and the Hamiltonian vector field, $X_{\hat H}$, of $\hat H$ on $\hat M$ under
\[ \pi_o: \hat M\to M\times\R^k \ni (m, a)\]
project to vector fields $\bar\D = \pi_{o*}\hat\D$ and $X = \pi_{o*} X_{\hat H}$ on this parameter space. Moreover,
\[ X(m,a) = X_{H_a}(m)\]
is $\bar\D$ invariant, ie $[X, \bar \D] = 0$, 
so projects to a vector field $\bar X$ on the quotient, $\bar M:= (M\times \R^k)/\bar\D$.

\begin{remark}\rm
The symmetry of $\bar\D$ on $X$ is easily understood. Under the assumptions of Proposition~\ref{thm:globalreduction}, 
if $\psi_s$ is the flow of $\D$ on $M$ then for $\gamma_a(t)\in M$ a trajectory of $X_{H_a}$, 
we have that $\psi_s(\gamma_a(t))$ is a trajectory of $H_{a'}$ where $a_j' = e^{(1-\Lambda_j)s}a_j$.
\end{remark}

From this last remark, the flow of $\bar X$ on $\bar M$ represents the trajectories of $H_a$ on $M$ modulo equivalence in parameter values 
and the scaling symmetry $\D$ on $M$. 
It is easy to see that these trajectories are all obtained by a projection of the contact flow of $\hat X$ on $\hat C$. 
Let us denote by 
\[\bar\pi_o: \hat C\to \bar M\]
the induced projection. Then, with the diagram
\begin{center}
\begin{tikzcd}
\hat M \arrow[r," \pi_o"] \arrow[d, "\hat\pi"] & M\times \R^k \arrow[d, "\bar\pi"] \\
\hat C \arrow[r, "\bar\pi_o"] & \bar M\\
\end{tikzcd}
\[\qquad \pi_{o*}X_{\hat H} = X ,~~ \pi_{o*}\hat \D = \bar\D, 
~~\hat C = \hat M / \hat\D ,~~\bar M = (M\times\R^k)/\bar\D\,,\]
\end{center}
we have:
\begin{theorem}\label{thm:lifted_red}\rm
Consider a Hamiltonian system depending on parameters as in Proposition~\ref{thm:globalreduction}. 
Let $\hat X := \hat\pi_* X_{\hat H}$ be the induced contact Hamiltonian vector field on $\hat C = \hat M/ \hat\D$. 
Then, for each fixed value $a_o$ of the parameters, there is an $\hat X$-invariant set $\hat C_{a_o}\subset \hat C$, with:
\[ \bar M_{a_o} := \bar \pi_o(\hat C_{a_o})\]
an $\bar X$-invariant set on $\bar M$, and moreover:
\[ \bar X|_{\bar M_{a_o}} = \bar\pi_{o*}\hat X.\]
	
	\end{theorem}

\begin{proof}
 Take $\hat C_{a_o} = \hat\pi( M\times\{ a_o\}\times \R^k)$.
\end{proof}

\begin{remark}\rm
As a special case of the above Proposition~\ref{thm:globalreduction} and Theorem~\ref{thm:lifted_red}, 
we may have all $\Lambda_j = \Lambda$, so that $\D$ is a scaling symmetry, of degree $\Lambda$, for each $H_a$ on $M$. 
In this case, there is a further projection from $\bar\pi_1:\hat C\to C = M/\D$, 
induced by $\pi_1:\hat M\to M$ so that the line fields $\ell_a$ on $C$ induced by $H_a$ are spanned by $\bar\pi_{1*}\text{span}\hat X|_{\hat C_{a_o}}$.
\end{remark}

	\begin{remark} \rm
	Some important observations here are in order:
	
	The first point is the fact that
    the assumption on $H_{a}$ in Proposition~\ref{thm:globalreduction} is not restrictive at all.
    Indeed, by introducing such parameters, basically any physical Hamiltonian can be written
    as a sum of pieces that scale differently.
    This allows one to apply contact reductions to 
    a much more general class of Hamiltonian systems than those which admit a %is much more general than focusing on Hamiltonian systems that admit a 
    scaling symmetry of some degree $\Lambda$,
    as the examples below shall show.
    
    The second crucial point
	(which might have interesting consequences on the way we formulate physical theories)
	is the fact that while on the lifted symplectic system $\hat H$ the dynamical variables $a$ are constant, on the 
	corresponding reduced contact system
	they become \emph{dissipated quantities}, 
	i.e.~they dissipate at the same rate as the contact Hamiltonian $\H$, 
	namely $L_{X_{\H}}a=-\mathscr{R}(\H)a$ (see~\cite{de2020infinitesimal,gaset2020new,bravetti2021geometric} to see
	that this is the equivalent of conserved quantities for the contact case).
	However, the global equivalence of the two dynamics is well 
	established by the discussion in Section~\ref{sec:contactreduction}.
	Therefore, we are ready to remark the significance of Proposition~\ref{thm:globalreduction} and Theorem~\ref{thm:lifted_red}:
    as anticipated in the discussion on Poincar\'e's dream in the Introduction, 
    since the parameters in a Hamiltonian description of the physical reality are \emph{inferred} from measurements
    based on the observed trajectories and on the specified model, there is no reason a priori to think of them as fixed numbers. Indeed
    we see from Theorem~\ref{thm:lifted_red} that if we allow the `space of physical Hamiltonian (variational) theories' describing a given
    phenomenon to vary, meaning that the coupling constants of each theory are different, then what we have just proved
    means that there always exists in this space a way to reduce the description to an equivalent theory based on a contact Hamiltonian
    (or Herglotz variational) theory on a reduced space.
    Despite the dynamical equivalence among all such theories, the contact Hamiltonian one, involving less elements for its complete description,
    has more
    descriptive power,
    %This has profound consequences from the perspective of the descriptive power of the theory, 
    as pointed out
    already in~\cite{gryb2021scale,sloan2021scale}. 
    We call this a \emph{reduction to a more descriptive theory}. 
    We note that the universal dissipative nature of the reduced system provides support for 
    the position that open systems are fundamental to physics~\cite{cuffaro2021open}. 
    
   A third point is that it is simple to extend Proposition~\ref{thm:globalreduction} 
    to construct a scaling symmetry of degree $\Lambda$ by taking 
    	$${\hat \D}:=\D+\sum_{i=1}^k{(\Lambda-\Lambda_i) a_{i}\partial_{a_{i}}}+\sum_{i=1}^k(1-\Lambda+\Lambda_i) b^{i}\partial_{b^{i}}\,.$$
    This construction may be of particular interest when investigating blow-ups of a particular degree; by choosing 
    $\Lambda$ such that the rescaling can tame the component of $H$ which becomes infinite, 
    we can then investigate the other components to see if a regularization can be achieved.
	\end{remark}

    Now we consider some simple examples that illustrate the generality of Proposition~\ref{thm:globalreduction}
    and Theorem~\ref{thm:lifted_red}.

    \begin{example}[Global reduction of Kepler]\rm\label{ex:scaledKglobal}
    Let us consider again the Kepler example (see Examples~\ref{ex:Kepler} and~\ref{ex:scaledK}).
      This time we reinsert the parameters $m$ (mass) and $k$ (Kepler coupling) 
    that we had previously fixed to 1 to obtain
    $$H_{m,k}=\underbrace{\frac{1}{m}\left(\frac{p_r^2}{2}+\frac{p_\theta^2}{2r^2}\right)}_{H_1}-\underbrace{\frac{k}{r}}_{H_2}\,.$$
    Clearly, for $\D=\D_K$, we have
    $$
    L_\D H_1 = -2 H_1 \qquad \text{and} \qquad L_\D H_2 = -2 H_2\,,
    $$
    and therefore we can apply the procedure detailed in Proposition~\ref{thm:globalreduction} to find the 
    degree one scaling symmetry
    $$\hat \D=\D_K+3\left(\mu\partial_{\mu}+k\partial_k\right)-2\left(b^\mu\partial_{b^\mu}+b^k\partial_{b^k}\right)\,,$$
    where $\mu:=m^{-1}$
    and, by Corollary~\ref{cor:D1}, the scale-reduced contact Hamiltonian dynamics on $\hat C=\hat M/\hat \D$.
    We stress that the final dynamics on $\hat C$ carries no information about the scale used for the reduction,
    and therefore needs one less parameter with respect to the original Kepler formulation. 
    Specifically, while the original problem needed $(r_0,\theta_0,p_r^0,p_\theta^0,m,k)$ as an input
    to completely specify a solution, the reduced problem only needs 
    e.g.~$(\theta_0,p_r^0,p_\theta^0,m_0,k_0)$ 
    (supposing we are using $r$ as the scaling function for the reduction).
    \end{example}
    
    Contrary to the Kepler example, in the next case we present a system that does not possess any evident scaling symmetry of degree $\Lambda$,
    but to which we can still apply Proposition~\ref{thm:globalreduction}.
    \begin{example}[Combined Kepler-Hooke System]\rm
    With $M$ and $\omega$ as before, consider now a Hamiltonian which combines the Kepler ($1/r$) and 
    the Hooke ($r^2$) potential:
    \[ H_{KH} := \underbrace{\mu\left(\frac{p_r^2}{2} + \frac{p_\theta^2}{2r^2}\right)}_{H_1}\,\, 
    \underbrace{-\frac{k_K}{r}}_{H_2}\,\,
    \underbrace{+k_H r^2}_{H_3}\,.\]
    Clearly in this case we do not have a (evident) scaling symmetry of degree $\Lambda$.
    However, we may notice that for $\D=-p_r \del_{p_r} + 2r \del_r + p_\theta \del_{p_\theta}$, we get
    \be\label{eq:KH1}
    L_\D H_1=-2H_1 \qquad  L_\D H_2=-2H_2 \qquad \text{and} \qquad L_\D H_3=4H_3\,.
    \ee
    Therefore, by applying Proposition~\ref{thm:globalreduction} we find the 
    degree one scaling symmetry 
    \be
    \hat \D=D+3\mu\partial_{\mu}+3k_K \partial_{k_K}-3k_H  \partial_{k_H}
    -2b^{\mu}\partial_{b^{\mu}}-2b^{k_K}\partial_{b^{k_K}}+4b^{k_H}\partial_{b^{k_H}}
    \ee
    and, by Corollary~\ref{cor:D1}, the scale-reduced contact Hamiltonian dynamics on $\hat C=\hat M/\hat \D$.
    \end{example}
    Indeed, we can generalize the last example at once to the case of any potential that can be written
    as a Laurent series of $r$ as follows.
     \begin{example}[General potential as a Laurent series of $r$]\rm
    Consider now the Laurent Hamiltonian
    \[ H_{L} := \underbrace{\mu\left(\frac{p_r^2}{2} + \frac{p_\theta^2}{2r^2}\right)}_{H_{\rm kin}} + 
    \underbrace{\sum_{j=-n_1}^{n_2} a_jr^j}_{\sum_{j=-n_1}^{n_2} a_jH_j}\,.\]
    Clearly in this case we do not have a (evident) scaling symmetry of degree $\Lambda$.
    However, we may notice that for $\D=-p_r \del_{p_r} + 2r \del_r + p_\theta \del_{p_\theta}$, we get
    \be\label{eq:KHgeneral}
    L_\D H_{\rm kin}=-2H_{\rm kin} \qquad \text{and} \qquad L_\D H_j=2jH_j\,.
    \ee
    Therefore, by applying Proposition~\ref{thm:globalreduction} we find the 
    degree one scaling symmetry 
    \be
    \hat \D=D+3\mu\partial_{\mu}+\sum_{j=-n_1}^{n_2}(1-2j)a_j \partial_{a_j}
    -2b^{\mu}\partial_{b^{\mu}}+\sum_{j=-n_1}^{n_2}2jb^j\partial_{b^{j}}
    \ee
    and, by Corollary~\ref{cor:D1}, the scale-reduced contact Hamiltonian dynamics on $\hat C=\hat M/\hat \D$.
    \end{example}
    
The next example is somewhat simpler, but it has important physical ramifications in cosmology, 
and therefore it is worth presenting it at this point.

\begin{example}[FLRW Cosmology]\label{ex:FLRW}\rm
In the field of cosmology, homogeneous, isotropic solutions to Einstein's equations are given by the Friedmann-Lema\^itre-Robertson-Walker metrics. 
The symmetries of such space-times allow us to reduce a full field theory to a finite-dimensional `particle' representation. 
In these solutions, the space-time manifold is $\mathbb{R} \times \Sigma_k$, 
in which $\Sigma_k$ is the three-dimensional spatial slice, 
consisting of a three sphere ($S^3, k=1$), Euclidean space ($\mathbb{R}^3, k=0$) or hyperbolic space ($\mathbb{H}, k=-1$). 
The metric is determined by the topology of the spatial section:
\[ ds^2 = -dt^2 + v^\frac{2}{3}(t) \left(\frac{dr^2}{1-kr^2} + r^2(d\theta^2 + \sin^2 \theta d\phi^2)\right). \]
In the case of a non-compact topology, the volume, $v$ is measured with respect to some fixed, fiducial cell. 
For minimally coupled matter, the dynamics of which are given in flat space by a Hamiltonian $H_m (p,q)$,
we find the equations of motion for the gravitational sector, $v,\Pi$, from the Hamiltonian:
\[ H = v \left( -\frac{3\Pi^2}{8\pi} + H_m \left(\frac{p}{v},q\right) -\frac{k}{v^\frac{2}{3}} \right). \] 
For flat $(k=0)$ spatial slices the Hamiltonian admits a degree one scaling symmetry $\D = v\partial_v + p\partial_p$ 
(summed over all matter momenta $p$ if multiple are present), with $i_\D \omega = v( -d\Pi + p dq)$. 
%The Hamiltonian, $H$, is a constraint - due to the reparametrization invariance of relativity, $H=0$. 
Therefore in this case the reduction to the contact dynamics on $C=M/\D$ is immediate, and
following Corollary~\ref{cor:D1} with the scaling function $\rho=v$, we have
\[ \mathscr{H}_0 = \frac{3 \Pi^2}{8\pi} - H_m(p,q) \qquad \text{and} \qquad \eta = -d\Pi + pdq .\]

To extend our results to the case in which $k \neq 0$, we follow the process outlined in Proposition~\ref{thm:globalreduction}
and Theorem~\ref{thm:lifted_red}.
Indeed, the Hamiltonian in this case can be split as
\[ H =  \underbrace{v \left( -\frac{3\Pi^2}{8\pi} + 
H_m \left(\frac{p}{v},q\right)\right)}_{a_1H_1}\,\,\underbrace{-k\,v^{1/3}}_{a_2H_2} \,, \] 
and for $\D = v\partial_v + p\partial_p$ as before we obtain
    $$L_\D H_1=H_1 \qquad \text{and} \qquad L_\D H_2=\frac{1}{3}H_2\,,$$
which, by Proposition~\ref{thm:globalreduction}, leads to the degree one scaling symmetry
$$\hat \D=\D + \frac{2}{3}k\partial_k+b^v\partial_{b^v}+\frac{1}{3}b^k\partial_{b^k}\,,$$
for the lifted system and, by Corollary~\ref{cor:D1}, to a scale-reduced contact Hamiltonian dynamics on $\hat C=\hat M/\hat \D$.

\begin{remark}\rm
There are two interesting facets of cosmological dynamics that are highlighted by this construction. 
The first is that the apparent `Hubble Friction' arising due to the expansion of the universe can be directly observed since 
the equation of motion for our system are given:
\[\dot{q} = \partial_p H_m \quad \dot{p} = -\partial_q H_m +\frac{3\Pi}{4\pi} p \,, \] 
with the effect of coupling $H_m$ to gravity in the final term acting as apparent friction 
when compared to the behaviour of matter in a non-expanding background. 
The second is that both the dynamics of matter in an expanding space, 
and the expansion of space that results from the presence of matter are described by taking 
a symplectic matter system $(M,\omega=d\theta, H_m)$ and considering on its contactification 
$(\tilde M,\tilde\theta) = (M\times \mathbb{R},\theta-d\Pi)$ the matter-gravity Hamiltonian $\tilde\H=\frac{3\Pi^2}{8\pi}-H_m$.
\end{remark}

\end{example}

\section{Conclusions}\label{sec:conclusions}

Let us  recapitulate the main points of our work. 
We have shown that a general class of scaling symmetries exist within symplectic systems, 
and how this symmetry can be used to reduce the system onto its invariants.
In Theorem~\ref{th:CHS} and its Corollary \ref{cor:D1} we have shown that the reduction by identifying orbits 
of the scaling vector field on the symplectic manifold yields a contact manifold, that the Hamiltonian vector field 
on the symplectic manifold projects onto a (contact) Hamiltonian vector field on the contact manifold, 
and that the contact Hamiltonian and contact form that generate this vector field are obtained by a simple procedure. 
We refer to this procedure as a \emph{contact reduction by scaling symmetries}, or simply \emph{contact reduction}.
This process has been further generalized beyond simple systems admitting a scaling symmetry. Indeed, by 
lifting systems with multiple couplings
we find that there is a general process under which these also admit a contact reduction to an equivalent
scale-invariant description
(Proposition~\ref{thm:globalreduction} and Theorem~\ref{thm:lifted_red}),
thus realizing
Poincar\'e's dream.

The resultant scale-reduced systems, being contact systems, have interesting features not apparent in the unreduced symplectic manifolds. 
The presence of a Reeb vector leads to measure focusing and dissipation of otherwise conserved quantities. 
Such features have important roles in describing physical phenomena such as the arrow of time \cite{gryb2021scale}. 

An interesting application of the contact reduction is to examine the nature of `blow-ups' (see Appendix~\ref{app:Kepler}),
for example, when
the scale parameter in the symplectic system approaches zero, such as at the total collisions of the three-body problem. 
Many of the problems regarding the continuation of solutions beyond such points arise due to the ill-defined evolution of the scale parameter. 
Such a situation can be avoided in the reduced system as the evolution of the scale is immaterial to the behaviour of the contact system. 
One can envision situations in which the contact Hamiltonian vector field will remain well-defined 
on the contact manifold when the symplectic Hamiltonian vector field becomes ill defined, and thus a `natural' continuation 
of the symplectic system may be considered as passing to the contact manifold for evolution beyond this point, 
and symplectifying the resultant evolution. In cosmological systems \cite{Koslowski:2016hds,Mercati:2019cbn,sloan2019scalar} 
this has been shown to allow continuation beyond the initial singularity, 
and in black holes beyond the central singularity~\cite{mercati2021through}. 

\appendix

\section{Symplectic and contact Hamiltonian mechanics}\label{app:ham_cont} 
In this appendix we review briefly some of the main definitions and properties of Hamiltonian mechanics on symplectic and contact manifolds, 
paying special attention to the relationships between the two that will be used in this work in order to have it self-contained.
Most of the definitions and results presented here can be found in~\cite{Arnold} and~\cite{MR2397738}, to which we refer for more detailed discussions.

\subsection{Symplectic and contact: definitions}
We begin with some general definitions that establish the notation and sign conventions.
	\begin{definition}\rm
	A \emph{symplectic manifold} is a pair $(M,\omega)$, where $M$ is a $2n$-dimensional manifold and $\omega$ is 
	a $2$-form on $M$ which is closed ($d\omega=0$) and non-degenerate ($\omega^{n}\neq 0$). A {\em symplectic vector field} on $M$ is one whose flow preserves $\omega$ : $L_X\omega = 0$. %($d\omega_{p}(v,w)=0$  $\forall w\in T_{p}M \implies v=0$).
	\end{definition}
Dynamics on a symplectic manifold is usually given by Hamiltonian systems, which are defined as follows.
	\begin{definition}\rm
	A (symplectic) \emph{Hamiltonian system} is a triple $(M,\omega,H)$, where $(M,\omega)$ is a symplectic manifold and $H:M\rightarrow\R$
	is a differentiable function called the (symplectic) \emph{Hamiltonian}. Given this structure, a unique symplectic vector field $X_{H}$, the \emph{(symplectic) Hamiltonian vector field},
	is defined on $M$ by the condition 
		\begin{equation}\label{eq:XHdef}
		i_{X_{H}}\omega=-dH\,.
		\end{equation}
	\end{definition}
	
By a theorem of Darboux, locally we can always find coordinates $(q^{a},p_{a})$, $a=1,\dots,n$, on a symplectic manifold 
such that
	\begin{equation}\label{eq:omegaDarboux}
	\omega=dp_{a}\wedge dq^{a}\,.
	\end{equation}	
These coordinates are called \emph{canonical} or \emph{Darboux} coordinates.
It follows directly from~\eqref{eq:XHdef} and~\eqref{eq:omegaDarboux} that the \emph{trajectories} of $X_{H}$
-- i.e.~parametrized curves $\gamma:I\subset \R\rightarrow M$ satisfying $\dot \gamma(t)=X_{H}(\gamma(t))$ --
are solutions to the \emph{symplectic Hamiltonian equations}
	\begin{equation}\label{eq:SHeq}
	\dot q^{a}=\frac{\partial H}{\partial p_{a}} \qquad \text{and} \qquad \dot p_{a}=-\frac{\partial H}{\partial q^{a}}\,.
	\end{equation}
%It should be remarked that $i_{X_{H}}\omega$ is always a closed $1$-form, but it can happen that it is not exact. Therefore Eq.~\eqref{eq:XHdef}

On the contact side, we have the corresponding definitions
%\subsection{Contact}
	\begin{definition}\rm
	A \emph{contact manifold} is a pair $(C,\mathscr D)$, where $C$ is a $(2n+1)$-dimensional manifold and $\mathscr D$ is a maximally
	non-integrable distribution of hyperplanes on $C$.
	Locally at least, there always exists a $1$-form $\eta$ on $C$ such that $\mathscr D=\ker\eta$
	and the maximal non-integrability condition is expressed as $\eta\wedge (d\eta)^{n}\neq 0$. 
	Such $\eta$ is called a \emph{contact 1-form}. A {\em contact vector field} on $C$ is one whose flow preserves $\mathscr{D}$ : $L_X\eta \sim \eta$.
	\end{definition}
	
Note that, a contact 1-form $\eta$ exists globally if and only if $\mathscr D$ is co-orientable, meaning that the quotient line bundle $TC/\mathscr D$ is trivial~\cite{MR2397738}. 
This will always be the case in this work.
Furthermore, once we fix a contact 1-form $\eta$ for $\mathscr D$, then any other 1-form of the type $\eta'=f\eta$ with $f:C\rightarrow \R^{\times}$
yields a different contact form for the same $\mathscr D$. That is, there is a freedom in choosing contact forms for the same $\mathscr D$, up to re-scaling
by a nowhere-vanishing function. However, to define certain dynamics, the choice of $\eta$ is relevant, as the next definition will show.

%Dynamics on a contact manifold is also given by Hamiltonian systems, which are defined as follows
%In particular, 
We recall that on a contact manifold with a fixed $\eta$ there is a special contact vector field, the \emph{Reeb vector field} $\mathscr R$, defined by
		\begin{equation}\label{eq:Reebdef}
		i_{\mathscr R}d\eta=0\qquad \text{and} \qquad i_{\mathscr R}\eta=1\,.
		\end{equation}
Now we are ready to introduce the Hamiltonian dynamics on contact manifolds 
(see as well Proposition~\ref{prop:arnoldSHVSCH} and Remark~\ref{rem:contRels} below for some context).
	\begin{definition}\rm
	A (contact) \emph{Hamiltonian system} is a triple $(C,\eta,\H)$, where $(C,\eta)$ is a contact manifold with a fixed choice for the contact 1-form 
	and $\H:C\rightarrow\R$
	is a differentiable function called the (contact) \emph{Hamiltonian}. Given this structure, a unique contact vector field $X_{\H}$, the \emph{contact Hamiltonian vector field},
	is defined on $C$, by the conditions 
		\begin{equation}\label{eq:XHcontactdef}
		i_{X_{\H}}d\eta=d\H-{\mathscr R}(\H)\eta\qquad \text{and} \qquad i_{X_{\H}}\eta=-\H\,.
		\end{equation}
	\end{definition}

By a theorem of Darboux, locally we can always find coordinates $(q^{a},p_{a},S)$ on a contact manifold 
such that
	\begin{equation}\label{eq:etaDarboux}
	\eta=dS-p_{a} dq^{a}\,, \qquad \text{and} \qquad \mathscr R=\frac{\partial}{\partial S}\,.
	\end{equation}	
These coordinates are called \emph{contact} or \emph{Darboux} coordinates.
It follows directly from~\eqref{eq:XHcontactdef} and~\eqref{eq:etaDarboux} that the trajectories of $X_{\H}$
are solutions to the \emph{contact Hamiltonian equations}
	\begin{equation}\label{eq:CHeq}
	\dot q^{a}=\frac{\partial \H}{\partial p_{a}} \qquad \dot p_{a}=-\frac{\partial \H}{\partial q^{a}}-p_{a}\frac{\partial \H}{\partial S} \qquad \dot S= p_{a}\frac{\partial \H}{\partial p_{a}}-\H\,.
	\end{equation}
It should be stressed at this point that the dynamical vector fields $X_{H}$ and $X_{\H}$ can be given variational definitions,
provided the corresponding Lagrangians are (hyper-)regular, 
both in the symplectic and in the contact case.
While in the symplectic case Hamilton's variational principle is well-known, 
its contact counterpart, known as Herglotz' variational principle, has been far less
studied until recently (see~\cite{georgieva2003generalized,vermeeren2019contact,de2019singular}). 
In Appendix~\ref{app:Herglotz} we detail this correspondence and we re-write the main results of this work from the variational perspective.

\subsection{Symplectic and contact: known relationships and new definitions}
Now we review some known relationships between symplectic and contact manifolds, together with the respective Hamiltonian dynamics.

We start with a procedure by which there always exists a natural extension of a contact manifold to a symplectic one %and it is always possible 
(see~\cite{Arnold}, Appendix 4).
	\begin{definition}\label{def:symplectification}\rm
	 Given a contact manifold, $(C,\mathscr{D})$, the \emph{symplectification}, $\tilde C$, of $C$ is:
	\begin{equation}\label{eq:symplectification}
	\tilde C:= \{ \beta_x\in T_x^*C : \ker\beta_x = \mathscr{D}_x\} \subset T^*C.
	\end{equation}  
	\end{definition}
Note that $\tilde C$ is an $\R^\times$ principal bundle over $C$ with (exact) symplectic structure $d\tilde\alpha$ defined by restriction of the canonical symplectic form on $T^*C$:
	\begin{equation}\label{eq:alpha}
		\tilde\alpha_{\tilde x}(\xi) := \tilde x((\pi_{C})_*\xi)\,,
	\end{equation} 
where $\pi_{C} :\tilde C\to C$ is the projection (see also~\cite{bruce2017remarks}).

\begin{remark}\rm\label{remark:symplect_components}
The symplectification of $C$ consists of the non-zero elements, $Ann(\mathscr{D})\backslash C$, of the annihilator of $\mathscr{D}$: $Ann(\mathscr{D}) := \{ \beta : \beta|_{\mathscr{D}} \equiv 0\}$ and so is naturally identified with the non-zero elements of $(TC/\mathscr{D})^*$. 

When the contact distribution is co-orientable ($TC/\mathscr{D}$ is a trivial line bundle), then so too is the symplectification a trivial principal bundle: $\tilde C \cong C\times \R^\times$, trivialized by a global choice of contact 1-form $\eta$. In this case $\tilde C$ has two connected components, $\tilde C_\pm$ and, when a contact 1-form $\eta$ has been chosen, it is convenient to refer to the component $\tilde C_+ := \{ e^s\eta\}$ consisting of positive multiples of $\eta$ as the symplectification.
\end{remark}

To understand further the relationship between $\tilde C$ and  $C$, we introduce the following definition. 
%also \emph{reduce} the symplectic manifold to a contact one.
%provide a converse of the above definition, we introduce the following
	\begin{definition}\label{def:LVF}\rm
	 Let $(M,\omega=d\lambda)$ be an exact symplectic manifold, where a choice for the symplectic potential $\lambda$ has been made.
	 The vector field $\X\in\mathfrak{X}(M)$ such that
	\begin{equation}\label{eq:Liouville}
	i_{\X}\omega=\lambda
	\end{equation}
%	$L_{\X}\omega=\omega$ (equivalently 
%	$i_{\X}\omega=\lambda+\alpha$, with $d\alpha=0$) 
	is called the \emph{Liouville vector field} of $\lambda$. 
%	Moreover, %given a Liouville vector field $\X$ on $M$,
%	any hypersurface $C$ transverse to $\X$ is called $\emph{of contact type}$.
	\end{definition}

Then, for $(\tilde C,\tilde\alpha)$ the symplectification of $(C,\eta)$, the generator, $\X$, of the $\R^\times$ action on $\tilde C$ is the  Liouville vector field of $\tilde\alpha$.
It turns out that there is a very special correspondence between contact vector fields on $C$ and certain Hamiltonian vector
fields on $\tilde C$.
To see this, observe that %the definition of contact Hamiltonian mechanics (see~\cite{Arnold}, Appendix 4): 
	symplectic flows of $\tilde X\in\mathfrak{X}(\tilde C)$ commuting with the flow of $\X$ induce flows of 
	$X = (\pi_{C})_*\tilde X\in\mathfrak{X}(C)$ preserving $\mathscr{D}$
	(contact flows).
This implies the following (see~\cite{Arnold}, Appendix 4)
%that there is an explicit relationship between contact Hamiltonians on $C$ and (some) Hamiltonian vector fields on $\tilde C$ (see~\cite{Arnold}, Appendix 4),
	\begin{proposition}\label{prop:arnoldSHVSCH}\rm
	Every contact vector field on $C$ can be lifted to a symplectic Hamiltonian vector field on $\tilde C$ (the symplectification of $C$)
	 whose Hamiltonian is homogeneous of degree 1
	with respect to $\X$, that is, $L_{\X}\tilde H= \tilde H$. %the action of $\R$, that is, $H(rx)=rH(x)$.
	Conversely, every Hamiltonian vector field on $\tilde C$ such that $\tilde H$ is homogeneous of degree 1 projects onto $C$ as a contact vector field. %$X_{\H}$.
	%In this case we call $\tilde H$ the \emph{symplectic lift of $\H$}.
	\end{proposition}
%every Hamiltonian $\H$ on $C$ corresponds to a Hamiltonian $\tilde H$ on $\tilde C$ that is homogeneous of degree one in the momenta (see~\cite{Arnold}, Appendix 4).
%and also~\cite{vanderschaft}). 
%In this case, we call $\tilde H$ the \emph{lift} of $\H$.

\begin{remark}\label{rem:contRels}\rm
  When a contact 1-form, $\eta$, of $C$ has been chosen, any contact vector field $X$ on $C$ may be associated with a function $\mathscr{H} = - i_X\eta$. 
  The vector field is recovered from the function (and contact 1-form), by the conditions $i_X\eta = - \mathscr{H}, ~ L_X\eta\sim \eta$, 
  which define a contact Hamiltonian vector field (eq.~\ref{eq:XHcontactdef}). 
  According to Proposition~\ref{prop:arnoldSHVSCH}, there is a Hamiltonian, $\tilde H$, on $\tilde C$ determined by $X$, 
  which we refer to as the {\em symplectic lift} of $\H$, related to $\H$ by $-\iota^*\tilde H = \H$, 
  for $\iota:C\to \tilde C$ the section associated to $\eta$.
\end{remark}

%In Darboux coordinates $(S, q^{a},p_{a})$ for $C$, one has 
%$\tilde C=C\times \R$, 
%$d\tilde\alpha=d(e^{s}\eta)$, 
%with $s$ the additional coordinate on $\R$. 
%Moreover, defining on $\tilde C$ 
%the coordinates 
%$Q^{\mu}=(Q^{0},Q^{a}):=(S,q^{a})$ and $P_{\mu}=(P_{0},P_{a}):=(e^s, -e^{s}p_{a})$, it is easy to check 
%that 
%$\tilde\alpha=+P_{\mu}dQ^{\mu}=e^{s}\eta$ -- i.e.~$(Q^{\mu},P_{\mu})$ 
%are Darboux coordinates on $\tilde C$ -- 
%and that the symplectic Hamiltonian vector field corresponding to
%	\begin{equation}\label{eq:relationshipHs}
%	\tilde H(Q^{\mu},P_{\mu}):= - P^0\H(Q_{0},Q^a,- P_{a}/P^{0})=Q^0\H(S,q^{a},p_{a})
%	\end{equation}
%projects onto $C$ precisely as the contact Hamiltonian vector field of $\H(S,q^{a},p_{a})$
See e.g.~\cite{Arnold} and~\cite{vanderschaft2018geomthermo,van2021liouville} for applications of this correspondence to thermodynamics. 
Above (Theorem~\ref{th:CHS} and Corollary~\ref{cor:coords}) we extend this relationship for Hamiltonians, 
$L_\X\tilde H = \Lambda\tilde H$ of degree $\Lambda$, on $\tilde C$. In fact, defining:

\begin{definition}\rm \label{def:LambdaHvf}
The {\em $\Lambda$-Hamiltonian vector field}, $X_{\mathscr{H}}^\Lambda$, 
of a function $\mathscr{H}$ on the contact manifold $C$ with respect to the contact 1-form $\eta$, is determined by the conditions:
\[ i_{X_\H^\Lambda}\eta = - \Lambda\H , ~~~ i_{X_\H^\Lambda}d\eta = d\H - (\mathscr{R}\H)~\eta.\]
\end{definition}

The results of Section~\ref{sec:symtocont} above may be stated in a more analogous manner to their degree one counterparts.

\begin{remark} \label{rem:LambdaHvf}\rm
 For $\Lambda = 1$, Definition~\ref{def:LambdaHvf} is exactly the definition of a contact Hamiltonian vector field. 
 The trajectories of a degree $\Lambda$ Hamiltonian, $\tilde H$, on $\tilde C$ project to (reparametrized) trajectories 
 of the $\Lambda$-Hamiltonian vector field on $C$ determined by $\mathscr{H} = - \iota^*\tilde H$, 
 where $\iota: C\to \tilde C$ is the section associated to $\eta$.
\end{remark}

%It is a simple exercise to extend this process to construct 
%a homogeneous Hamiltonian of degree $\Lambda$ on $\tilde C$,
%by multiplying 
%$\tilde H$ by $(Q^0)^{\Lambda-1}$. However, in such case the dynamics of the resulting symplectic Hamiltonian does not
%project onto the dynamics of $\H$ (see Theorem~\ref{th:CHS} and Corollary~\ref{cor:coords} for their precise relationship).

As a completely analogous construction to the symplectification, 
if $(M,\omega)$ is exact symplectic, there is also a natural procedure to {extend} it to %a symplectic manifold $(M,\omega)$ to 
a contact manifold
%and it is possible as long as $\omega$ is exact 
(see~\cite{Arnold}, Appendix 4).
	\begin{definition}\label{def:contactification}\rm
	 Given an exact symplectic manifold $(M,\omega=d\lambda)$, the \emph{contactification} of $M$ is 
	\begin{equation}\label{eq:contactification}
	\tilde M:=M\times \R, \qquad \tilde{\lambda}:=dS-\pi_{M}^{*}\lambda\,,
	\end{equation}   
	where $\pi_{M}:\tilde M\rightarrow M$ is the standard projection.
	\end{definition}
Again, there is an explicit relationship between Hamiltonian vector fields $M$ and certain contact vector fields on $\tilde M$, given by
    \begin{proposition}\label{prop:arnoldCHVSSH}\rm
	Every symplectic Hamiltonian vector field on $M$ can be lifted to a contact Hamiltonian vector field on $\tilde M$ (the contactification of $M$)
	 whose Hamiltonian is 
	 	\begin{equation}\label{eq:contactHlift}	
		 \tilde \H=\pi_{M}^{*}H\,.
	 	\end{equation}
Conversely, every contact Hamiltonian vector field on $\tilde M$ such that $\mathscr{\tilde R}\tilde\H=0$ projects onto $M$ as a symplectic Hamiltonian vector field $X_{H}$.
	In this case we call $\tilde \H$ the \emph{contact lift of $H$}.
	\end{proposition}
In Darboux coordinates $(q^{a},p_{a})$ for $M$, one has $\tilde\lambda=dS-p_{a}dq^{a}$
and the integral curves of $X_{\tilde \H}$ are given by~\eqref{eq:CHeq} by setting $\tilde \H(q^a,p_a,S):=H(q^a,p_a)$ (note that 
in particular $\partial \tilde\H/\partial S=0$).

Note that, by definition, the flow of $X_{H}$ is just the projection of that of $X_{\tilde \H}$.
However, on $\tilde M$ there can be defined contact Hamiltonians that are not lifts of any symplectic Hamiltonian on $M.$ 
This has been
exploited largely in recent years to describe mechanical systems with dissipation (classical and quantum), thermostatted systems and 
thermomechanical phenomena in a Hamiltonian framework (see e.g.~\cite{grmela1997dynamics,grmela2014contact,bravetti2015contact,bravetti2016thermostat,bravetti2018contact,ciaglia2018contact,bravetti2021geometric,goto2016contact,goto2021nonequilibrium,entov2021contact}).

Dual to the above constructions, in which we \emph{extended} either a contact or a symplectic manifold, under the appropriate hypotheses
one can also \emph{reduce} one of these structures to the other.
In particular, we have
	\begin{definition}\label{def:ofCtype}\rm
	 Let $(M,\omega=d\lambda)$ be an exact symplectic manifold and $\X\in\mathfrak{X}(M)$
	its \emph{Liouville vector field}. %Moreover, %given a Liouville vector field $\X$ on $M$,
	Any hypersurface $C$ transverse to $\X$ is called $\emph{of contact type}$.
	\end{definition}
Hypersurfaces of contact type on an exact symplectic manifold are naturally endowed with a contact structure as follows.
	\begin{proposition}\label{prop:contacttype}\rm
	Let $(M,\omega=d\lambda)$ be an exact symplectic manifold, $\X$ the Liouville vector field and $\iota_{C}:C\rightarrow M$ a hypersurface of contact type.
	Then the 1-form $\eta_{\X}:=\iota_{C}^{*}i_{\X}\omega=\iota_{C}^{*}\lambda$ is a contact form on $C$. %, where $i:C\rightarrow M$ is the canonical inclusion.
	\end{proposition}
		\begin{example}[Transverse $H$ levels]\label{ex:H=E}\rm
	A typical example of a hypersurface of contact type is given by any regular energy level $C_{E}:=H^{-1}(E)$ of a Hamiltonian system 
	on the cotangent bundle $(T^{*}Q,d\alpha)$, where $\alpha$ is the canonical 1-form,
	which is transverse to the Liouville vector field. 
	Interestingly, in this case the restriction of the Hamiltonian vector field $X_{H}$ to $C_{E}$ %$H^{-1}(E)$ 
	is just a reparametrization
	of the Reeb vector field of $\eta_{\X}$. Indeed, one has
%		\begin{equation}
	$\mathscr{R}=\frac{X_{H|C_{E}}}{\X(H)_{|C_{E}}},$
%		\end{equation}
	and therefore the \emph{orbits} -- the unparametrized trajectories -- of the two vector fields are the same. %~\cite{bravetti2020invariant}.	
	%, i.e.~hypersurfaces	such that $H=E$, for some fixed value $E\in\R$. 
	\end{example}

Analougously, one can reduce a contact manifold to an exact symplectic one on a proper ``energy hypersurface'' of the contact Hamiltonian 
and find an equivalence between 
the orbits of the contact Hamiltonian vector field and the Liouville one as follows (see~\cite{bravetti2020invariant}).
	\begin{definition}\label{def:ofStype}\rm
	 Let $(C,\eta)$ be a contact manifold and $\mathscr R\in\mathfrak{X}(C)$
	its \emph{Reeb vector field}. %Moreover, %given a Liouville vector field $\X$ on $M$,
	Any hypersurface $S$ transverse to $\mathscr R$ is called $\emph{of symplectic type}$.
	\end{definition}
Hypersurfaces of symplectic type on a contact manifold are naturally endowed with a symplectic structure as follows.
	\begin{proposition}\label{prop:symplectictype}\rm
	Let $(C,\eta)$ be a contact manifold, $\mathscr R$ the Reeb vector field and $\iota_{S}:S\rightarrow C$ a hypersurface of symplectic type.
	Then the 2-form $\Omega=d\theta$, with $\theta = \iota_S^*\eta$ is an exact symplectic form on $S$. %, where $i:C\rightarrow M$ is the canonical inclusion.
	\end{proposition}

\begin{example}[Transverse $\H$ 0-level]\label{ex:cH=E}\rm
Let $\H: C \to \R$ be a Hamiltonian function on a contact manifold $(C, \eta)$ with Reeb vector field $\mathscr R$ 
and assume that $S = \H^{-1}(0) \not=\emptyset$ and that $\mathscr R(\H)(x)\not=0,$ for all $x\in S.$ Then $(S,\Omega)$
is an exact symplectic manifold as above. Moreover, 
if $\Delta$ is the Liouville vector field of the exact symplectic manifold $(S,\Omega)$, %i.e.~$\Delta$ is the vector field 
then $X_{\H|S}$ is the reparametrization of $\Delta$ given by 
%\begin{equation}\label{X-H-S-Delta}
$X_{\H|S}=-(\mathscr R(\H)\circ \iota_S)\Delta.$
%\end{equation}
%\end{enumerate}
\end{example}

Finally, let us also note that all the above relationships have motivated the definition of a \emph{symplectic sandwich with contact bread}, see~\cite{bravetti2020invariant},
where it was proved that the existence of an invariant measure for $X_{\H|S}$ is equivalent to the existence of a symplectic sandwich with contact bread. %\davescomment{Feel free to flesh this out a bit more - nice result should be broadcast!}

\section{Herglotz' Principle}\label{app:Herglotz}

In this appendix, we will consider certain contact reductions
admitting variational descriptions (details will be left to the reader). 
We refer to eg~\cite{georgieva2003generalized,vermeeren2019contact,de2019singular} for the variational formulation of contact
systems and to~\cite{wang2016implicit,wang2019aubry,liu2018contact,liu2021orbital} 
for the use of  variational techniques in order to study dynamical properties and solutions of the related Hamilton-Jacobi equations.

Let us recall that certain Hamiltonian systems admit Lagrangian variational descriptions.
	\begin{definition}\rm
	A \emph{Lagrangian system} is a pair $(Q,L)$, where $L:TQ\rightarrow \R$. The {\em action} of a curve $q(t)\in Q$ is $S:= \int L(q, \dot q)~dt$.
	\end{definition}

	%\begin{definition}\rm
	%Let $\Omega$ be the space of curves $\gamma:[0,1]\rightarrow Q$. An \emph{action functional} for a Lagrangian system $(Q,L)$ is a map $S:\Omega \rightarrow \R$ given $S(\gamma)=\int_\gamma L $.
	%\end{definition}
	
The extremals of the action, with say fixed endpoints, satisfy the well-known Euler-Lagrange equations:
\begin{equation}
 \frac{d}{dt}\left(\frac{\partial L}{\partial \dot{q}} \right) = \frac{\partial L}{\partial q}
\end{equation}

%From these definitions we find that extremizing the action over the space of curves with fixed endpoints (i.e. $\gamma(0)=s,\gamma(1)=f$ for some fixed $s$ and $f$ in $Q$), gives rise to the well-known Euler Lagrange equations:
%\begin{equation}
% \frac{d}{dt}\left(\frac{\partial L}{\partial \dot{q}} \right) - \frac{\partial L}{\partial q} = 0
%\end{equation}

%Given a Lagrangian $L'$ there exists a 1-form on $TQ$ given $\lambda=\frac{\partial L'}{\partial \dot{q}} dq$. In our constructions, the 2-form $d\lambda$ will play a parallel role to that of the symplectic structure in the Hamiltonian constructions. There are two useful constructions to note at this point. The first is that, subject to the equations of motion $\dot{\lambda}=dL'$. The second is that, for any coordinate $q \in Q$, $\iota_{\partial_q} \lambda = \frac{\partial L'}{\partial \dot{q}}$.

The contact analogue of a Lagrangian system is
 
 	\begin{definition}\rm
	A \emph{Herglotz Lagrangian system} (also known as a \emph{contact Lagrangian system}) is a pair $(\bar Q,\mathscr{L})$, where $\mathscr{L}:T\bar Q \times \R \rightarrow \R$. The {\em action} of a curve 
	$q(t)\in \bar Q$ 
	is $\Delta S :=\int \mathscr{L}(q,\dot{q},S) dt$, where $S(t)$ solves $\dot S = \mathscr{L}(q, \dot q, S)$.
	\end{definition}

%This definition extends the scope of Lagrangian systems to allow for dependence on the action, $S$, and can be equivalently described through the differential equation 
%\begin{equation} 
%\dot{S} = L(q,\dot{q},S)
%\end{equation}
 This definition extends the scope of Lagrangian systems to allow for dependence on the action. 
 The extremals of this action, with say fixed endpoints and initial condition $S(0)=S_0$, 
 satisfy the Herglotz-Lagrange equations:
 \begin{equation}
     \frac{d}{dt}\left(\frac{\partial \mathscr{L}}{\partial \dot{q}}\right) =  \frac{\partial \mathscr{L}}{\partial S}\frac{\partial \mathscr{L}}{\partial \dot{q}} + \frac{\partial \mathscr{L}}{\partial q} , \qquad \dot S = \mathscr{L}.
 \end{equation}
 
 The connection between Herglotz Lagrangian systems and contact Hamiltonian systems runs in parallel to the usual connection between Lagrangian and Hamiltonian systems. Under the {\em Legendre transform},
 \[ Leg: T\bar Q\times\R \to T^*\bar Q\times\R , \qquad  p = \frac{\del \mathscr{L}}{\del \dot q}, \]
  the Herglotz-Lagrangian extremals correspond to trajectories of the contact Hamiltonian system
  \begin{equation}
 C=T^*\bar Q \times \R, \qquad  \H = p \cdot \dot{q} - \mathscr{L}, \qquad \eta = dS - p\cdot dq.
 \end{equation}
 
 %Identifying $p \in T^*Q$ by $p = \frac{\partial L}{\partial \dot{q}}$ we make the Legendre transform exchanging $(q,\dot{q},S)$ for $(q,p,S)$, and arrive at the contact Hamiltonian system: $(C,\eta,\H)$ given
 %\begin{equation}
 %C=T^*Q \times \R \quad  \H = p \dot{q} - L \quad \eta = dS - pdq
 %\end{equation}
 %whose dynamics are equivalent to those generated by the Herglotz Lagrangian system. 
 
 \begin{remark}\rm
     Here we consider systems whose Legendre transform is a diffeomorphism. One way to see the connection of a contact Hamiltonian system to its Herglotz variational principle is via the Poincar\'e-Cartan form: $\alpha := p\cdot dq - \mathscr{H}~dt$ and $\hat \eta := dS - \alpha = \eta + \mathscr{H}~dt$ on the extended space, $C\times\R\ni (c, t)$. The world lines of trajectories of the contact Hamiltonian system are characterized by lying tangent to the line field $\ker\hat\eta \cap \{X: i_Xd\alpha\sim \eta\}$, or equivalently as extremals of $\gamma\mapsto \int_\gamma \alpha$ among curves satisfying the constraint $\hat\eta|_{\gamma}\equiv 0$.
 \end{remark}
 
 ~~One may obtain `Lagrangian analogues' of our results above by applying the Legendre transform to 
 Hamiltonian systems and their contact reductions. 
 It is instructive to first consider a local coordinate analogue of Proposition~\ref{prop:arnoldSHVSCH}, or Corollary~\ref{cor:D1} (the degree one case).
 
 \begin{example}\label{ex:deg1LagRed}\rm
    Consider a Herglotz-Lagrangian system $\mathscr{L}(q, \dot q, S)$ corresponding under Legendre transform to the contact Hamiltonian system $\mathscr{H}(q,p,S),~\eta = dS - p\cdot dq$. The contact Hamiltonian system is the scale reduction of the Hamiltonian system:
    \[ H(\rho, q, S, P) = - \rho\mathscr{H}(q, -P/\rho, S)\]
    on the symplectification, $\omega = d\tilde\alpha = d (\rho~dS + P\cdot dq)$, by the degree one scaling symmetry $\D = \rho\del_\rho + P\cdot \del_P$. We identify this symplectification with the cotangent bundle of $ Q\ni (\rho, q)$ by taking $S = - p_\rho$, with canonical 1-form:
    \[ \lambda_{o} = p_\rho~d\rho + P\cdot dq = \tilde\alpha + d(\rho p_\rho).\]
    For which $S = - \frac{i_\D\lambda_o}{\rho}$. Now, the Legendre transform of $H$ is the Lagrangian system $L = - (\dot\rho S + \rho \mathscr{L})$, or
    \[ \mathscr{L} = - \frac{ L + \dot\rho S }{\rho}\]
    is our Herglotz-Lagrangian system on $q\in \bar Q =  Q/\rho\del_\rho$, for the scale reduction of the Lagrangian system $L$ by the scaling symmetry $\D_L = \rho\del_\rho + \dot\rho\del_{\dot\rho}$ (corresponding to $\D$ under Legendre transform).
    We note that we have as well the equivalent coordinate expressions:
    \[\mathscr{L} = - \del_\rho L ,~~ S = - \del_{\dot\rho}L.\]
 \end{example}
 
%\begin{proposition} \label{prop:HLEL} 
 
% Suppose $(Q,L)$ is a Herglotz Lagrangian system, with Darboux coordinates $(S,q,\dot{q})$. Let $x=-\int^t \frac{\partial L}{\partial S}$, and $\pi:Q' \rightarrow Q$ be the identity map on $Q$. Then $(Q',L')=(Q \times \R, e^x(L+S\dot{x}))$ is a Lagrangian system such that the Euler-Lagrange equations on $Q'$ are equivalent to the Herglotz-Lagrange equations on $Q$ and the role of $S$ as action for $L$.
 
 %\end{proposition}
 
% This can be shown by direct calculation of the Euler-Lagrange equations for our $L'$. The Euler-Lagrange equations for $L'$ are:
% \begin{eqnarray}
 %            0:&=& \frac{d}{dt}\left(e^x \frac{\partial L}{\partial \dot{q}}\right) - e^x \frac{\partial L}{\partial q}  = e^x \left( \frac{d}{dt} \left(\frac{\partial L}{\partial \dot{q}} \right) -\frac{\partial L}{\partial S}\frac{\partial L}{\partial\dot{q}}- \frac{\partial L}{\partial q} \right) \nonumber \\
  %           \frac{dS}{dt} &=& L 
 %\end{eqnarray}
 
 Only certain Hamiltonian systems on cotangent bundles admit Lagrangian variational descriptions, 
 and as well only certain contact Hamiltonian systems admit Herglotz variational descriptions. 
 We will consider the following special case of scale reduction.
 
 \begin{proposition}[Cotangent contact reduction] \rm \label{prop:cotRed}
  Let $H, T^*Q$ be a Hamiltonian system admitting a scaling symmetry, $\D$, of degree $\Lambda$ such that
  \[   L_{\D}\lambda_o = \lambda_o\]
  where $\lambda_o$ is the canonical 1-form on $T^*Q$ ($\D$ being a \emph{basic scaling symmetry}). 
  Then $\D$ projects to a vector field, $\bar\D$, on $Q$. Assume that the flow of $\bar \D$ acts freely and properly on $Q$, and set
  \[ \bar Q := Q/\bar\D.\]
  Then, given a scaling function $\rho:Q\to\R^\times$ of $\bar \D$, one may identify:
  \[ C = (T^*Q)/\D \cong T^*\bar Q\times\R \ni (\bar q, \bar p , S)\]
  with contact 1-form $\eta = dS - \bar\lambda_o$, where $\bar\lambda_o$ is the canonical 1-form on $T^*\bar Q$, and $S = - i_\D\lambda_o/\rho$.
 \end{proposition}
 
% Here we note that the Lagrangian in proposition \ref{prop:HLEL} can be arrived at by considering the Herglotz Lagrangian, finding its contact Hamiltonian system, symplectifying that contact Hamiltonian system following \ref{prop:arnoldSHVSCH} then following a Legendre transform to return the equivalent Lagrangian. The relationship is laid out in the following commutative diagram:

The description of a scale reduction of a Lagrangian system may be obtained by applying the Legendre transform 
to a Hamiltonian system and its contact reduction of the form in Proposition~\ref{prop:cotRed}. 
The relationship is laid out in the following commutative diagram:
\[ \begin{tikzcd}
TQ, L \arrow[leftrightarrow]{r}{Legendre} \arrow[swap]{d}{Quotient} & T^*Q, H \arrow{d}{Quotient} \\%
T\bar Q\times \R, \mathscr{L}~~ \arrow[leftrightarrow]{r}{Legendre}& ~~T^*\bar Q\times\R, \mathscr{H}
\end{tikzcd}
\]

Thus, on the Lagrangian side, we define:

\begin{definition} \rm
    A \emph{basic scaling symmetry} of {\em degree} $\Lambda$ for a Lagrangian system $(Q, L)$ is a vector field $\D$ on $TQ$ such that 
    $\D(L)= \Lambda L$ and $L_{\D} \lambda_L = \lambda_L$, where $\lambda_L = \del_{\dot q} L\cdot dq$ 
    is the {\em Lagrangian 1-form} (the pullback of the canonical 1-form by the Legendre transform).
\end{definition}

As we have remarked above (Remark~\ref{rem:LambdaHvf}), the trajectories of a degree $\Lambda$ Hamiltonian system scale reduce to trajectories of a $\Lambda$-Hamiltonian vector field on $C$. Such trajectories may be described variationally.

\begin{definition} \rm
    A {\em $\Lambda$-Herglotz system} is a pair, $(\bar Q, \mathscr{L})$, 
    where $\mathscr{L}: T\bar Q\times\R\to \R$. 
    The action of a curve $q(t)\in\bar Q$ is $S_\Lambda := \int \mathscr{L}(q, q', S)~dt$, where $S$ solves $S' = \mathscr{L} + (1-\Lambda)\mathscr{E}$, for $\mathscr{E} := Leg^*\mathscr{H} = \del_{q'}\mathscr{L}\cdot q' - \mathscr{L}$.
\end{definition}

\begin{remark}\rm
    Under Legendre transform, $\mathscr{H} = p\cdot \dot q - \mathscr{L}$, 
    extremals of the $\Lambda$-Herglotz system $\mathscr{L}$ correspond to trajectories of the $\Lambda$-Hamiltonian vector field of $\mathscr{H}$ (Definition~\ref{def:LambdaHvf}). 
    One way to see this connection is via the Poincar\'e-Cartan form: 
    $\alpha := p\cdot dq - \mathscr{H}~dt$ and $\hat \eta_\Lambda := \eta + \Lambda\mathscr{H}\,dt$ on the extended space, $C\times\R\ni (c, t)$. 
    The world lines of trajectories of the $\Lambda$-Hamiltonian vector field 
    are characterized by lying tangent to the line field $\ker\hat\eta_\Lambda \cap \{X: i_Xd\alpha \sim \hat\eta_\Lambda\}$, 
    or equivalently as extremals of $\gamma\mapsto \int_\gamma \alpha$ among curves satisfying 
    the constraint $\hat\eta_\Lambda|_{\gamma}\equiv 0$. 
    Alternately, one may verify the correspondence directly via the equations of motion (eqs.~\eqref{eq:contCoords} above). 
\end{remark}

%Following from the above definitions, $\iota_D \dot{\lambda} = \Lambda L'$, hence subject to the equations of motion $\iota_D \lambda = \Lambda (S' -S'_0)$. 
%Note that the case where $\pi_{Q'}(D)=0$ is special as it does not affect a solution curve $\gamma \in Q'$, but rather can only serve to change the time parametrization. We discard such transformations from our further analysis. In other cases, we can separate $Q'$ into the kernel of $D$ and a coordinate $x$ such that $\pi_{Q'}(D)=\partial_x$. 
%If $\Lambda=1$ then the transformation is isochronal, and thus preserves the Hamiltonian $H=\dot{q} \frac{\partial L'}{\partial \dot{q}} - L'$. In such cases we note that we can choose $v=e^x$ and thus on $TQ$, $D=v\partial_v + \dot{v}\partial_{\dot{v}}$. We will use this fact directly in what follows.

With this definition, we have:

\begin{theorem}[Counterpart to Theorem~\ref{th:CHS}] \label{thm:LagReduce} \rm
Let $(Q,L)$ be a Lagrangian system, admitting a basic scaling symmetry, $\D$, of degree $\Lambda$ 
and a scaling function $\rho:Q\to \R^\times$. Then one may identify 
\[\pi: TQ\to (TQ)/\D \cong T\bar Q \times\R \ni (\bar q, \bar q', S)\]
where $\pi^*S = - i_{\D}\lambda_L /\rho$. The extremals of the Lagrangian system project to extremals of the $\Lambda$-Herglotz system:
\[\pi^*\mathscr{L} = - \frac{L + \dot\rho S}{\rho^\Lambda}.\]
\end{theorem}
%$(Q,L)=(\frac{Q'}{\pi_Q(D)},\frac{\partial L'}{\partial v})$ is a Herglotz Lagrangian system wherein $S=\frac{\partial L'}{\partial \dot{v}}$, and the dynamics of $L$ and $L'$ coincide on $Q$.

%\begin{proof}
%First note that $\frac{d}{dt}\left(\frac{\partial L'}{\partial \dot{v}}\right) = \frac{\partial L'}{\partial v} = L $. 
%Hence $S$ is an action for $L$. 
%Then note that since $\Lie_D L'=L'$ then 
%$L'=v\frac{\partial L'}{\partial v}+\dot{v} \frac{\partial L'}{\partial \dot{v}} = vL + \dot{v}S$. 
%Identifying $v$ with $e^x$ in proposition \ref{prop:HLEL}, we see that the dynamics are equivalent on $Q'$.
%\end{proof}

\begin{remark}\rm
 As in the Hamiltonian case (see Remark~\ref{rem:coordinates} above), 
 the scale-reduced trajectories are reparametrized according to \[\rho^{1-\Lambda}d\tau = dt,\]
 and we write $' = \frac{d}{d\tau}$. 
 Note for $\Lambda = 1$, the Lagrangian system scale reduces to a Herglotz-Lagrangian system, 
 as we saw in Example~\ref{ex:deg1LagRed} above.
\end{remark}

\begin{example}[The 2d harmonic oscillator]\label{ex:LagHooke} \rm
Let us consider the two dimensional harmonic oscillator of Example~\ref{ex:2DHO}, with Lagrangian
\begin{equation}
    L = \frac{\dot{r}^2 + r^2 \dot{\theta}^2 - k r^2}{2},
\end{equation}
in polar coordinates. 
We have the basic scaling symmetry: $\D=\frac{1}{2}(r \partial_r + \dot{r} \partial_{\dot{r}})$, of degree one, 
and scaling function $\rho=r^2$. Following Theorem~\ref{thm:LagReduce}, we write %Our Lagrangian and scaling function can thus be rewritten
\begin{equation}
    S = -\frac{i_\D\lambda_L}{\rho} = -\frac{\dot\rho}{4\rho} , \qquad L = \frac{-\dot\rho S + \rho\dot\theta^2 - k\rho}{2}, %\frac{\dot{\rho}^2}{8\rho} + \frac{\rho \dot{\theta}^2}{2} - \frac{k_H \rho}{2} \quad D = \rho \partial_\rho + \dot{\rho} \partial_{\dot{\rho}}
\end{equation}
to obtain the scale-reduced Herglotz Lagrangian system: 
\begin{equation}
    \mathscr{L} = -\frac{L + \dot\rho S}{\rho} = 2S^2  +  \frac{k - \dot{\theta}^2}{2} %\frac{\partial L'}{\partial \rho} = -2S^2 + \frac{\dot{\theta}^2}{2} - \frac{k_H}{2}
\end{equation}
on $TS^1\times\R\ni (\theta, \dot\theta, S)$. 
Note that $\mathscr{L}$ is the (contact) Legendre transform of $\mathscr{H} = - 2S^2 - \frac{k  + \bar p_\theta^2}{2}$ on $T^*S^1\times \R$ 
obtained by contact reduction of $H = \frac{|p|^2 + k|q|^2}{2}$ by the scaling symmetry $\frac{q\cdot \del_q + p\cdot\del_p}{2}$ and $\rho = |q|^2$.
\end{example}

%\begin{remark}
 %    Consider a scaling symmetry that acts on a configuration variable and time: $\underline{D}: Q\times \mathbb{R} \rightarrow Q \times \mathbb{R}$ such that $D$ maps minima of an action onto one-another. Then the tangent lift of $\underline{D}$ on $TQ$ can be made to take the form $v\partial_v + v'\partial_{v'}$ by a choice of variable on configuration space, and a time reparametrization. On the zero energy surface this will map between solutions to the Euler-Lagrange equations, and is isochronal in the new reparametrized time.
%\end{remark}

\begin{example}[Contact-reduced Kepler] \label{ex:LagK}\rm
We consider the Kepler system of Example~\ref{ex:Kep_scal}, with Lagrangian:
\[ L = \frac{|\dot q|^2}{2} + \frac{1}{|q|}\]
and basic scaling symmetry $\D_K = 2q\cdot \del_q - \dot q\cdot \del_{\dot q}$, of degree $\Lambda = -2$, and scaling function $\rho = |q|^{1/2}$. 
Following Theorem~\ref{thm:LagReduce}, we take polar coordinates and write:
\[ S = - \frac{i_\D\lambda_L}{\rho} = - 4\rho^2\dot\rho, \qquad L = \frac{-\dot\rho S + \rho^4 \dot\theta^2}{2} + \frac{1}{\rho^2}\]
to obtain the scale-reduced $-2$-Herglotz system:
\[ \mathscr{L} = - \rho^2( L + \dot\rho S) = \frac{S^2}{8} - \frac{(\theta')^2}{2} - 1\]
on $TS^1\times\R\ni (\theta, \theta', S)$ (using $'$ to denote $\frac{d}{d\tau}$ where $\rho^3 d\tau = dt$). 
Again, note that $\mathscr{L}$ is related by Legendre transform to the $-2$-Hamiltonian system, 
$\mathscr{H}$, obtained by an analogous contact reduction of the Legendre transform of $L$. 
We also observe that the reparametrization $\rho^3d\tau = dt$ has:
\[L~dt = \rho \left(\frac{2\rho'^2}{\rho^2} +\frac{\theta'^2}{2} +1 \right) d\tau = \tilde L~d\tau\]
where $\mathscr{L} = -\frac{\tilde L + \rho'S}{\rho}$ gives the scale-reduced $-2$-Herglotz system ($S = -4\frac{\rho'}{\rho}$).
    %For the Kepler problem, the scaling symmetry is $\underline{D}=2r\partial_r + 3t\partial_t$. Hence choosing $v = \sqrt{r}$ and $d\tau = v^{-3} dt$ we can express the Lagrangian density as 
    %\begin{equation}
    %    L dt = v\left(\frac{2v'^2}{v^2} +\frac{\theta'^2}{2} +1 \right) d\tau        
    %\end{equation}
    %for which it is clear that the tangent lift $D$ of $\underline{D}$ takes the required form.
\end{example}

Finally, we consider a counterpart to Theorem~\ref{thm:lifted_red} 
for certain Lagrangian systems depending on parameters, $a_j$, of the form:
\begin{equation}\label{eq:NatLagCoupled}
    L_a =  T(q, \dot q) + \sum_{j=1}^k a_j L_j(q)
\end{equation}
and for which we have a vector field $\D$ on $TQ$ with
\begin{equation}\label{eq:NatLagSS}
    {\D}(T) = T ,  \qquad {\D}(L_j) = \Lambda_j L_j , \qquad L_{\D}\lambda_L = \lambda_L\,.
\end{equation}
Note that the first condition in~\eqref{eq:NatLagSS} 
fixes the parametrization of $\D$, being the precise analogue of the first condition in Definition~\ref{def:scaling}.
\begin{theorem}[Counterpart to Theorem~\ref{thm:lifted_red}]\label{thm2:laglifted}\rm
Consider a Lagrangian system depending on parameters as in eq.~\eqref{eq:NatLagCoupled}, 
and admitting a vector field $\D$ as in eq.~\eqref{eq:NatLagSS}. 
Then the Lagrangian system
\[ \hat L := T +  \sum_{j=1}^k  \dot X_j \left(\frac{L_j}{\dot X_j}\right)^{\frac{1}{\Lambda_j}} \]
on $\hat Q := Q\times\R^k\ni (q, X)$ admits $\D$ as a basic degree one scaling symmetry. 
Moreover, the projection to $Q$ of an extremal of $\hat L$ is an extremal of $L_a$, with parameter values
\[ \Lambda_j a_j = \left( \frac{\dot X_j}{L_j}\right)^{1 - \frac{1}{\Lambda_j}}\]
constant over the extremals of $\hat L$.
%$L'=\sum_{i=1}^{k} L'_i$ where $\Lie_{\tilde{D}} L'_i = \Lambda_i L'_i$. Then there exists a Lagrangian system $(\mathfrak{Q},\mathfrak{L})$ for which $\tilde{D}$ is a Lagrangian scaling symmetry of degree 1, which generates the same equation of motion on $Q'$. Herein
%$$\mathfrak{L} = \sum_{i=1}^k \dot{X}_i^\frac{\Lambda_i-1}{\Lambda_i} {L'_i}^\frac{1}{\Lambda_i} \quad \mathfrak{Q}=Q' \times \R^k $$
%with boundary condition $\dot{X}_i = \Lambda_i^\frac{\Lambda_i}{\Lambda_i-1} L'_i$. 
\end{theorem}

%\begin{proof}
% That $\tilde{D}$ is a Lagrangian scaling symmetry of degree 1 follows directly from the form of $\mathfrak{L}$. Let us consider $q_j$ to be coordinates on $Q'$. Applying the Euler-Lagrange equations to $\mathfrak{L}$ we find:
% \begin{eqnarray}
% \frac{d}{dt}\left(\frac{\partial \mathfrak{L}}{\partial \dot{X}_i}\right)  = 0 \rightarrow \dot{X}_i = C_i L'_i \\
% \frac{\partial \mathfrak{L}}{\partial \dot{q}_j} = \sum_{i=1}^k \frac{\dot{X_i}^{\frac{\Lambda_i-1}{\Lambda_i}}}{\Lambda_i} {L'_i}^{\frac{1-\Lambda_i}{\Lambda_i}} \frac{\partial {L'_i}}{\partial \dot{q}_j} = \sum_{i=1}^k \frac{\partial L'_i}{\partial \dot{q}_j} = \frac{\partial L'}{\partial \dot{q}_j}
%  \\
%\frac{\partial \mathfrak{L}}{\partial q_j} = \sum_{i=1}^k \frac{\dot{X_i}^{\frac{\Lambda_i-1}{\Lambda_i}}}{\Lambda_i} {L'_i}^{\frac{1-\Lambda_i}{\Lambda_i}} \frac{\partial {L'_i}}{\partial q_j} = \sum_{i=1}^k \frac{\partial L'_i}{\partial q_j} = \frac{\partial L'}{\partial q_j}
% \end{eqnarray}
% Hence the Euler-Lagrange equations from $\mathfrak{L}$ for the $q_j$ are equivalent to those from $L'$. 
%\end{proof}

\begin{remark} \rm More generally, for $L = \sum a_j L_j$, one may consider $\hat L = \sum \dot X_j \left( \frac{L_j}{\dot X_j}\right)^{1/\Lambda_j}$.
Under Legendre transform, this lifted system, $\hat L$, 
corresponds to the Hamiltonian system $\hat H = \sum a_j H_j$, 
where $\Lambda_j a_j = \left(\frac{\dot X_j}{L_j}\right)^{1 - \frac{1}{\Lambda_j}}$ and $H_j = \del_{\dot q}L_j\cdot \dot q - L_j$ 
is the Legendre tranform of $L_j$. 
A Legendre transform of $\hat \D$ from Remark~\ref{remark:coupledHamv2}, yields a degree one scaling symmetry for this Lagrangian system, $\hat L$ 
(which in general is {\em not} $\D$).

Similar comments as those of Remark~\ref{remark:coupledHamv2} apply here, 
in that in order to avoid potential complex numbers from the exponents $\frac{1}{\Lambda_j}$, 
one may consider various cases with absolute values to take care of the signs. 
In general, to avoid such exponents, one may consider eg:
\[ \hat L := \sum e^{\dot X_j + L_j}\]
as the lift of the Lagrangian system, $L_a = \sum a_j L_j$, where $e^{\dot X_j + L_j} = a_j$ is constant over extremals of $\hat L$, with the degree one scaling symmetry $\hat \D = \sum \del_{\dot X_j}$.
%In parallel to the construction of theorem 2, we have found that by lifting our Lagrangian on $Q$ we can form a system that has a scaling symmetry of degree 1. This can then be combined with our earlier proposition to see that a Herglotz Lagrangian system can be formed on the smaller space.
\end{remark}

\begin{example}[The Kepler-Hooke potential]\label{ex:LagKeplerHooke}\rm
    Let us consider the system described in Example~\ref{ex:LagHooke} and extend the system to include a Kepler potential:
    \begin{equation}
    L = \underbrace{ \frac{\dot{\rho}^2}{8\rho} + \frac{\rho \dot{\theta}^2}{2} - \frac{k_H \rho}{2}}_{L_1} +  
    \underbrace{\frac{k_K}{\sqrt{\rho}}}_{L_{2}}  \qquad \D = \rho \partial_\rho + \dot{\rho} \partial_{\dot{\rho}}
    \end{equation}
    For this system, $\D$ is not a scaling symmetry, as $\D(L_1)=L_1$ but $\D(L_2)=-\frac{1}{2} L_2$. 
    However, following Theorem~\ref{thm2:laglifted} we can consider 
    \begin{equation} \label{eq:laglifted}
        \hat L =  \frac{\dot{\rho}^2}{8\rho} + \frac{\rho \dot{\theta}^2}{2} - \frac{k_H \rho}{2} + \frac{\dot{X}^3 \rho}{k_K^2} 
    \end{equation}
    for which $\D$ is a scaling symmetry. 
    We then impose the initial condition $\dot{X}_0 = \frac{k_K}{(-2)^{\frac{1}{3}}r_0}$, where $r_0$ is the initial value of $r$. 
    This gives rise to the same equations of motion as the original system. 
    Since this is a Lagrangian scaling system of degree one, we can follow Example~\ref{ex:LagHooke} and find the equivalent Herglotz Lagrangian system:
    \begin{equation} \label{eq:lagliftedreduced}
        \mathscr{L} = -2S^2 + \frac{\dot{\theta}^2}{2} - \frac{k_H}{2} + \frac{\dot{X}^3}{k_K^2}\,.
    \end{equation}
\end{example}

\begin{remark}\rm
    Note that we could have dropped $k_K$ entirely from the description of the system in equations~\eqref{eq:laglifted} and~\eqref{eq:lagliftedreduced} 
    by appropriately choosing a boundary condition for $\dot{X}$, 
    and hence we have in essence exchanged specifying a coupling constant, $k_K$, for an initial condition. 
\end{remark}

\section{Blow-ups in celestial mechanics}\label{app:Kepler}

In this appendix, we remark on the contact reduction by scaling symmetries presented here in relation to some well-known constructions 
in celestial mechanics (see for example~\cite{Chenciner1997alinfini,miranda2018geometry,MontgomeryBlowup}).

Let us first recall the scale-reduced Kepler problem (Examples~\ref{ex:Kepler}--\ref{ex:Kblowup}). 
We denote the scaling functions of Example~\ref{ex:Kep_scal} by
\[ \rho := |q|^{1/2} ,~~ \kappa:=\frac{1}{|p|} ,~~ G:= p\cdot iq, ~~J:= p\cdot q.\]
where $\rho^4 = q\cdot q$ is the moment of inertia, and $\frac{1}{\kappa^2} = p\cdot p$ is (twice) the kinetic energy.
    
    Of course, the contact reductions corresponding to each choice of scaling function represent the same curves in $C = M/\D$ 
    obtained by projections of Kepler orbits. They are merely different choices of coordinates for $C$. 
    However, each choice of coordinates highlights certain properties of these scale-reduced curves. 
    For instance, the scaling function $\rho$, is naturally associated to a blow-up of the collision orbits (Example~\ref{ex:Kblowup}), 
    while the scaling function $\kappa$ to a blow-up of orbits with $|p|\to\infty$ (as well the collision orbits in the Kepler problem).
    
To determine some contact reductions (Table~\ref{tab:Kep_reds} below) for various choices of scaling function, 
we first note that the scaling functions above are not independent, satisfying the relation $\kappa^2 (J^2 + G^2) = \rho^4$.
    One may use any three of these scaling functions, along with a scale invariant angle, as coordinates on the phase space. There are essentially three (dependent) natural angles, $\theta, \fe, \delta = \fe - \theta$, present,
    where
    \begin{equation}\label{eq:Kep_angles}
        q = \rho^2 e^{i\theta} ,    \qquad p = \frac{1}{\kappa} e^{i\fe} , \qquad \rho^2 e^{i\delta} = \kappa (J + iG)
    \end{equation}
    for $\delta$ the angle between position and momentum.
    %, and we may write
    %\[ \iota^2 p = (J + iG) e^{i\theta}\]
    
    To see the equations of motion obtained by these various choices of scaling functions, 
    we will consider $(\rho, \theta, J, G)$ for coordinates on the phase space. Then we have (eq.~\eqref{eq:KepContRed} above): 
    \[ \lambda =  G~d\theta + 2J~ d\log\left( \frac{\rho}{J}\right), \qquad 
    H_K = \frac{1}{2\kappa^2} - \frac{1}{\rho^2} ,
    \qquad 
    \kappa^2(J^2 + G^2) = \rho^4.\] %G~d\theta + 2\left( J~d(\log\iota) - dJ\right) 
    
    One may describe the contact reduction for various choices of scaling functions according to Remark~\ref{remark:restr} 
    by restriction of $\lambda$ and $-H_K$ to a level set of the scaling function, 
    obtaining the following representations of the scale-reduced Kepler problem:
    \begin{table}[h!]
        \centering
        \begin{tabular}{||c|c|c||}
        \hline
          Scaling function   &  Contact system & Scale invariant equations of motion\\
          \hline\hline
           1.~~$\rho$  & $\eta = \tilde G~d\theta - 2~d\tilde J ,~~\mathscr{H}_0 = 1 - \frac{\tilde J^2 + \tilde G^2}{2}$ & $\tilde J' = \frac{\tilde G^2}{2} - \mathscr{H}_0,~~\tilde G' = - \frac{\tilde J\tilde G}{2} ,~~ \theta' = \tilde G$\\
           \hline
           2.~~$\kappa$ & $\eta_\kappa = G_\kappa~d\fe - dJ_\kappa ,~~\mathscr{H}_\kappa = \frac{1}{R} -  \frac{1}{2}$ & $J_\kappa' = - 2\mathscr{H}_\kappa + \frac{G_\kappa^2}{R^3},~~G_\kappa' = - \frac{G_\kappa J_\kappa}{R^3} ,~~\fe' = \frac{G_\kappa}{R^3}$\\
           \hline
           3.~~$G$ & $\eta_G = d\theta - J_G~d \log\left(\frac{J_G}{\rho_G}\right)^2 ,~~\mathscr{H}_G = \frac{2\rho_G^2  - J_G^2 - 1}{2\rho_G^4}$ & $J_G' = \frac{1}{\rho_G^2} - 2\mathscr{H}_G ,~~\rho_G' = \frac{J_G}{2\rho_G^3} ,~~\theta' = \frac{1}{\rho_G^4}$\\
           \hline
           4.~~$J$ & $\eta_J = G_J~d\theta + d \log\rho_J^2 ,~~\mathscr{H}_J = \frac{2\rho_J^2 - G_J^2 - 1}{2\rho_J^4}$ & $G_J' = G_J(2\mathscr{H}_J - \frac{1}{\rho_J^2}) ,~~\rho_J' = \rho_J\mathscr{H}_J - \frac{G_J^2}{2\rho_J^3} ,~~\theta' = \frac{G_J}{\rho_J^4}$\\
           \hline
        \end{tabular}
        \caption{\small{Contact reductions of the Kepler problem with various scaling functions. In row 2 we set $R := \sqrt{J_\kappa^2 + G_\kappa^2}$.}}
        \label{tab:Kep_reds}
    \end{table}
    the contact reduction with $\rho$ (item 1 of Table~\ref{tab:Kep_reds}) is Example~\ref{ex:Kep_scal} above. 
    In this table we used, for example, in the $\kappa$ case, the notation:
    \[ \pi^*J_\kappa = J/\kappa ,\quad\pi^*G_\kappa = G/\kappa ,\quad\pi^*\fe = \fe ,\quad \pi^*\mathscr{H}_\kappa = -\kappa^2 H_K ,\quad \pi^*\eta_\kappa = \lambda/\kappa\,,\]
    the subscript notation in the other cases taking the analogous meaning.

    \begin{remark}\rm
        We have already seen (Example~\ref{ex:Kblowup}), the relation of row 1 in Table~\ref{tab:Kep_reds} to the blow-up of the collision orbits. 
        Likewise, row 2 yields a blow-up of $|p|\to \infty$ orbits. 
        Here, for the Kepler problem, these are as well the collision orbits. 
        In general $n$-body problems, choices of various mutual distances as scaling functions or analogues of $\kappa$ will be interesting to examine. 
        
        We also remark that the use of a first integral as a scaling function, as in row 3, corresponds to a first integral for the scale-invariant curves. 
        Here, this first integral, $\pi^*\mathscr{H}_G = - G^2 H_K$, is a well known scale invariant (sometimes referred to as the Dziobek constant). 
        Also observe that the equations of motion for the scale invariant $J$ are forms of the well-known Lagrange-Jacobi identity.
        
        The reduction based on the scaling function $J$ of row 4, is used in Barbour et.~al's~\cite{barbour2014identification} (with $n$-body problems), 
        where it is referred to as the dilational momentum and used, for instance, to describe long term `clustering' behavior of solutions.
    \end{remark}

    \begin{remark}\rm
        One has an analogous contact reduction for certain inverse power force laws:
        \[H_\alpha = \frac{|p|^2}{2} - \frac{1}{\alpha |q|^\alpha}\]
        with $\D_\alpha:= \frac{2q\cdot \del_q - \alpha p\cdot \del_p}{2-\alpha}$ a scaling symmetry of degree $\Lambda = -\frac{2\alpha}{2-\alpha}$ 
        (when $\alpha\ne 2$). 
        Note that $J, G$ above are still scaling functions, as is $\rho_\alpha := |q|^{1 - \alpha/2}$. 
        Carrying out the analogue of Exampple~\ref{ex:Kblowup} here gives exactly the blown-up collision tori 
        (see \S 1.3 of \cite{devaneySingCM}) of these central force problems.
    \end{remark}
    
    The contact reduction may be carried out similarly for $n$-body problems:
    \[ q = (q_1,...,q_n) \in \R^{nd} ,\qquad p = (p_1,...,p_n) \in \R^{nd}\]
    with $q_j, p_j\in \R^d$ given as the Hamiltonian system:
    \[ \omega = dp\wedge dq = d(p\cdot dq) ,\qquad  H = \frac{\|p\|^2}{2} - U(q)\]
    where $U(\lambda q) = \lambda\inv U(q)$ is homogeneous of degree $-1$. 
    For simplicity consider the masses are unit (otherwise, one may take, $\|~\|$, as a suitable mass-weighted norm). 
    The system admits the degree $-2$ scaling symmetry:
    \[ \D = 2q\cdot \del_q - p\cdot\del_p.\]
    
    Taking $\rho := \|q\|^{1/2}$ as a scaling function, the quotient $\R^{2nd} / \D \cong S^{nd - 1}\times \R^{nd}$, has scale-invariant coordinates:
    \[s := q/\rho^2 \in S^{nd - 1} ,\qquad y := \rho p \in \R^{nd}.\]
    
    The contact reduction associated to the scaling function $\rho$ gives the well-known McGehee blow-up. 
    Taking $\rho \eta = i_\D\omega = -2q\cdot dp - p\cdot dq$ and $\mathscr{H} = - \rho^2 H$, one has:
    \[  \eta = d\nu + s\cdot dy ,  \qquad \mathscr{H} = U(s) - \frac{\|y\|^2}{2}\]
    where $\nu := s\cdot y$. 
    The scale-invariant equations of motion (eqs.~\eqref{eq:contCoords}, with $\Lambda = - 2$) are exactly the usual 
    McGehee blow-up equations (eg \S 4.2 of~\cite{MontgomeryBlowup}):
    \[ s' = y - \nu s ,\qquad y' = \frac12 \nu y + \nabla U(s)\,,\]
    the collision manifold being given by the invariant level $\mathscr{H} = 0$ (see as well \cite{mercati2021total}).

 %------------------------------------------------------------------------------------------------------------------------------------------------------
%------------------------------------------------------------------------------------------------------------------------------------------------------
%------------------------------------------------------------------------------------------------------------------------------------------------------

\bibliographystyle{abbrvnat_mv}
\bibliography{contact.bib}

%------------------------------------------------------------------------------------------------------------------------------------------------------
%------------------------------------------------------------------------------------------------------------------------------------------------------
%------------------------------------------------------------------------------------------------------------------------------------------------------

\end{document}